\documentclass[12pt]{extarticle}
\usepackage{amsmath, amsthm, amssymb, color}
\usepackage{graphicx}
\usepackage{caption}
\usepackage{subcaption}
\usepackage{diagbox}
\usepackage[shortlabels]{enumitem}
\usepackage{makecell}
\usepackage[version=4]{mhchem} 

\usepackage[ruled,vlined,linesnumbered]{algorithm2e}
\usepackage{expl3}

\usepackage{dsfont}

\usepackage{tikz}
\usepackage{pgfplots} 
\usepgfplotslibrary{fillbetween}
\pgfplotsset{compat=1.14}
\usepackage{comment}
\usetikzlibrary{arrows,positioning,matrix,decorations.pathmorphing,calc} 

\usepackage{booktabs}

\usepackage{lipsum}

\usepackage[affil-it]{authblk}

\oddsidemargin \evensidemargin
\theoremstyle{plain}
\newtheorem{theorem}{Theorem}[section]
\newtheorem{corollary}[theorem]{Corollary}
\newtheorem{proposition}[theorem]{Proposition}
\newtheorem{lemma}[theorem]{Lemma}
\newtheorem{conjecture}[theorem]{Conjecture} 
\newtheorem{question}[theorem]{Question}

\theoremstyle{definition}
\newtheorem{defn}[theorem]{Definition}

\newtheorem{notation}[theorem]{Notation}

\newtheorem{remark}[theorem]{Remark}

\newcommand{\scc}{\mathcal{S}}
\newcommand{\F}{f_{c, \kappa}}
\newcommand{\augmentH}{h_{c, a}}

\newcommand{\St}{{S}}

\usepackage[colorlinks=true, allcolors=blue]{hyperref} 

\newcommand{\newt}{\mathrm{Newt}}

\newcommand{\argmin}{\mathrm{argmin}}
\newcommand{\vol}{\mathrm{Vol}}

\newcommand{\sign}{\text{sign}}
\newcommand{\R}{\mathbb{R}}
\newcommand{\C}{\mathbb{C}}

\newcommand{\koff}{k_{\mathrm{off}}}
\newcommand{\kon}{k_{\mathrm{on}}}
\newcommand{\kcat}{k_{\mathrm{cat}}}
\newcommand{\lcat}{{\ell}_{\mathrm{cat}}}
\newcommand{\loff}{{\ell}_{\mathrm{off}}}
\newcommand{\lon}{{\ell}_{\mathrm{on}}}
\newcommand{\stot}{\text{S}_{\text{tot}}}
\newcommand{\etot}{\text{E}_{\text{tot}}}
\newcommand{\ftot}{\text{F}_{\text{tot}}}

\usepackage{mathtools}

\newcommand{\defword}[1]{\textcolor{purple}{\underline{#1}}}

\begin{document}

\title{Oscillations and bistability in a model of ERK regulation}

\date{March 6, 2019}

\author[1]{Nida Obatake} 
\author[1]{Anne Shiu}
\author[1]{Xiaoxian Tang}
\author[2]{Ang\'elica Torres}
\affil[1]{Department of Mathematics, Texas A\&M University, USA}
\affil[2]{Department of Mathematical Sciences, University of Copenhagen, Denmark}

\maketitle

\begin{abstract}
This work concerns the question of how two important dynamical properties, oscillations and bistability, emerge in an important biological signaling network.
Specifically, we consider a model for dual-site phosphorylation and dephosphorylation of extracellular signal-regulated kinase (ERK).  We prove that oscillations persist even as the model is greatly simplified (reactions are made irreversible and intermediates are removed).  Bistability, however, is much less robust -- this property is lost when intermediates are removed or even when all reactions are made irreversible.  Moreover, bistability is characterized by the presence of two reversible, catalytic reactions: 
as other reactions are made irreversible, 
bistability persists as long as one or both of the specified reactions is preserved.  
Finally, we investigate the maximum number of steady states, aided by a network's ``mixed volume'' (a concept from convex geometry).  Taken together, our results shed light on the question of how oscillations and bistability emerge from a limiting network of the ERK network -- namely, the fully processive dual-site network -- which is known to be globally stable and therefore lack both oscillations and bistability.
Our proofs are enabled by a Hopf bifurcation criterion due to Yang, 
analyses of Newton polytopes arising from Hurwitz determinants, 
and recent characterizations of multistationarity for networks having a steady-state parametrization.

  \vskip 0.1cm
  \noindent \textbf{Keywords:} chemical reaction network, Hopf bifurcation, oscillation, bistable, Newton polytope, mixed volume  

\end{abstract}

\section{Introduction} 
In recent years, significant attention has been devoted to the question of how 
bistability and oscillations  
emerge in biological networks involving multisite phosphorylation~\cite{cs-survey}. 
Such networks are of great biological importance~\cite{25}. 
The one we consider is the network, depicted in Figure~\ref{fig:erk_sys}, comprising extracellular signal-regulated kinase (ERK) regulation by dual-site phosphorylation by the kinase MEK 
(denoted by $E$) and dephosphorylation by the phosphatase MKP3 ($F$)~\cite{long-term}. 
This network, which we call the ERK network, 
has an important role in regulating many cellular activities, with dysregulation implicated in many cancers~\cite{mek-erk}.  Accordingly, 
an important problem is to understand the dynamical 
properties of the ERK network, with the goal of predicting 
effects arising from mutations or drug treatments~\cite{erk-model}.  

\begin{figure}[th]
    \centering
   \includegraphics[scale=1.2]{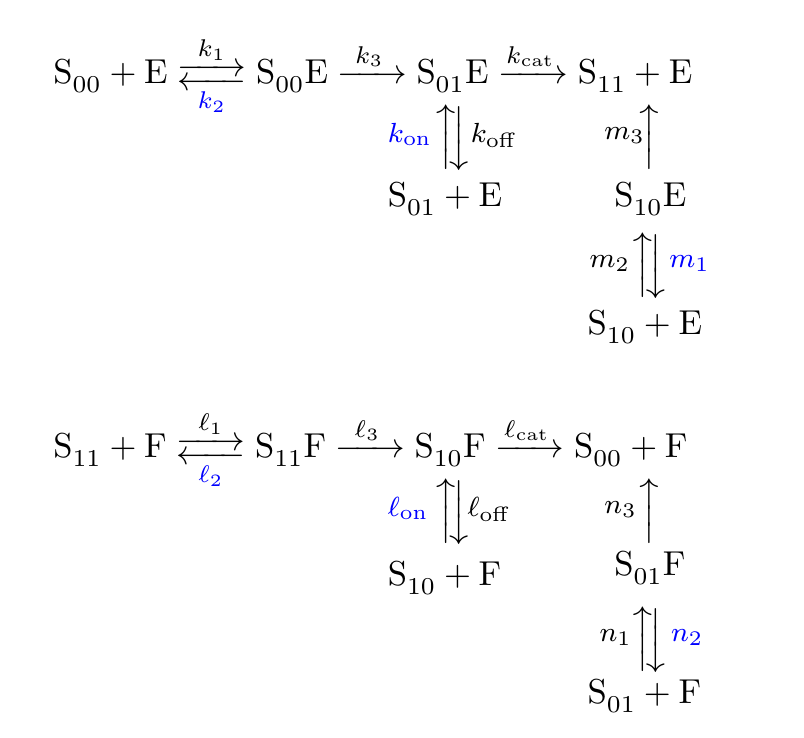} 
   \caption{The (full) \defword{ERK network}, from~\cite{long-term}, with notation of~\cite{DPST}.  The {\em fully processive} network is obtained by deleting all vertical reactions (those labeled by $\kon, \koff, m_1, m_2, m_3, \lon, \loff, n_1, n_2, n_3$).  We also consider irreversible versions of the ERK network obtained by deleting some of the reactions labeled $k_2,\kon, m_1, l_2, \lon, n_2$ (in blue).  In particular, deleting all six of those reactions yields the \defword{fully irreversible ERK network}.}
   \label{fig:erk_sys}
\end{figure}

The ERK network was shown by 
Rubinstein, Mattingly, Berezhkovskii, and Shvartsman~\cite{long-term}
to be bistable and exhibit oscillations (for some choices of rate constants).
Rubinstein {\em et al.}\ 
also observed that 
the ERK network ``limits'' to a network without bistability or oscillations.
Namely, when the rate constants $\kcat$ and $\lcat$ are much larger than $\koff$ and $\loff$, respectively, 
this yields the ``fully processive'' network obtained by deleting all vertical arrows in Figure~\ref{fig:erk_sys}, which is globally convergent to a unique steady state~\cite{ConradiShiu,EithunShiu}. 
Accordingly, Rubinstein {\em et al.}\ asked, How do bistability and oscillations in the ERK network emerge from the processive limit? This question was subsequently articulated as follows \cite{cs-survey}:

\begin{question} \label{q:main}
When the processivity levels $p_k:={\kcat}/{(\kcat+\koff)}$ and $p_{\ell}:={\lcat}/{(\lcat+\loff)}$ are arbitrarily close
to 1, is the ERK network still bistable 
and oscillatory?
\end{question}

One of our main contributions is to lay foundation toward answering Question~\ref{q:main}.  Speci\-fi\-cally, we answer a related question, {\em How do bistability and oscillations emerge from simpler versions of the ERK network?}  
Our main results, summarized in Table~\ref{tab:summary},
are that oscillations are surprisingly robust to operations that simplify the network, while bistability is 
lost more easily. 
Specifically, oscillations persist even as reactions are made irreversible and intermediates are removed (see Section~\ref{sec:oscillations}),
while bistability is lost more quickly, when only a few reactions are made irreversible (Section~\ref{sec:bistability}).  
Taken together, our results form a case study for the problem of model choice -- 
an investigation into the simplifications of a model that preserve important dynamical properties.

\begin{table}[ht]
\begin{center}
\begin{tabular}{lcccc}
\hline
 ERK & & & Maximum \# & Maximum \# \\
network          & Oscillations?                     & Bistability?                      & 
steady states & over $\mathbb C$ \\
\hline
Full         & Yes~\cite{long-term} & Yes~\cite{long-term} & Conjecture: 3                        & 7                            \\
 Irreversible & Yes* & If $\kon >0$ or $\lon >0$ & 1 & 5*\\
Reduced     & Yes                     & No                                & 1                                    & 3    \\                       
\hline
\end{tabular}
\caption{Summary of results.   Yes* indicates that the fully irreversible ERK network exhibits oscillations (see Figure~\ref{fig:oscill-in-irr-erk}), and 5* indicates that 5 is the maximum number of complex-number steady states for the network obtained 
from the full ERK network 
by setting $\kon=0$.
For details on results, see Propositions~\ref{prop:hopfcriterion}, \ref{prop:bistable}, and~\ref{prop:mixed-vol}, and Theorem~\ref{thm:bistable}.
\label{tab:summary}}
\end{center}
\end{table}

Our focus here -- on determining which operations on the ERK network preserve oscillations and bistability -- is similar in spirit to the recent approach of Sadeghimanesh and Feliu~\cite{core}.  Indeed, there has been significant interest in understanding which operations on networks preserve oscillations~\cite{banaji-inheritance}, bistability~\cite{BP-inher, FeliuWiuf, Joshi-Shiu-2013}, and other properties~\cite{joining}.

A related topic -- mentioned earlier -- is the question of how dynamical properties arise in phosphorylation systems.  Several works have examined this problem at the level of parameters, focusing on the question of which rate constants and/or initial conditions give rise to oscillations~\cite{mixed} or bistability~\cite{CFMW,conradi-iosif-kahle}.  Our perspective is slightly different; instead of allowing parameter values to change, we modify the network itself.  
Accordingly, our work is similar in spirit to recent investigations into minimal oscillatory or bistable networks~\cite{banaji-inheritance, BP-inher, hadac-osc, single-mol, which-small}.

A key tool we use is a parametrization of the steady states.  Such parametrizations have been shown in recent years to be indispensable for analyzing multistationarity (multiple steady states, which are necessary for bistability) and oscillations~\cite{  Giaroli-Bihan-Dickenstein,johnston-mueller-pantea,ThomGuna09}.
Indeed, here we build on results in~\cite{CFMW,mixed, DPST}.

Specifically, following~\cite{mixed}, we investigate oscillations by employing a steady-state parametrization together with a criterion of Yang~\cite{yang-hopf} that characterizes Hopf bifurcations in terms of determinants of Hurwitz matrices.
In~\cite{mixed}, this approach showed that the Hopf bifurcations of a mixed-mechanism phosphorylation network lie on a hypersurface defined by the vanishing of a single Hurwitz determinant.  For our ERK networks, however, 
the problem does not reduce to the analysis of a single polynomial, and the size of these polynomials makes the system difficult to solve. 
To this end, we introduce an algorithm for analyzing these polynomials, through their Newton polytopes, by using techniques from polyhedral geometry.  Using this algorithm, we succeed in finding, for the reduced ERK network, a Hopf bifurcation giving rise to oscillations.


Finally, we investigate the precise number of steady states in ERK networks.  
For general networks, much has been done for determining which networks admit multiple steady states -- see e.g.~\cite{CFMW, ME_entrapped, DPST, FeinOsc, Giaroli-Bihan-Dickenstein, mss-review, signs} -- but there are few techniques for determining a network's maximum number of steady states.  To this end, we introduce two related measures of a network, the maximum number of complex-number steady states and the ``mixed volume''.  In general, the mixed volume is an upper bound on the number of complex-number steady states, but we show that these numbers are equal for ERK networks (Section~\ref{sec:max-mumber}). 

The outline of our work is as follows.
Section~\ref{sec:bkrd} contains background on chemical reaction systems, 
steady-state parametrizations, and Hopf bifurcations.
We present steady-state parametrizations for the ERK network and the reduced ERK network in Section~\ref{sec:parametrizations}.
Section~\ref{sec:results} contains our main results on oscillations and bistability.
Section~\ref{sec:max-mumber} investigates the number of steady states and the relationship to mixed volumes. 
We end with a Discussion in Section~\ref{sec:discussion}. 

\section{Background} \label{sec:bkrd}
Here we introduce chemical reaction systems (Section~\ref{sec:CRS}),
their steady-state parametrizations (Section~\ref{sec:param}), and 
Hopf bifurcations (Section~\ref{sec:hopf}).

\subsection{Chemical reaction systems} \label{sec:CRS}
As in~\cite{DPST}, our notation closely matches that of Conradi, Feliu, Mincheva, and Wiuf~\cite{CFMW}. 
A \defword{reaction network} $G$  (or {\em network} for short) comprises a set of $s$ species $\{X_1, X_2, \ldots, X_s\}$ and a set of $m$ reactions:
\[
\alpha_{1j}X_1 + 
\alpha_{2j}X_2 +  \dots +
\alpha_{sj}X_s
~ \to ~
\beta_{1j}X_1 + 
\beta_{2j}X_2 +  \dots +
\beta_{sj}X_s~,
 \quad \quad
    {\rm for}~
	j=1,2, \ldots, m~,
\]
where each $\alpha_{ij}$ and $\beta_{ij}$ is a non-negative integer. The \defword{stoichiometric matrix} of 
$G$, 
denoted by $N$, is the $s\times m$ matrix with
$(i, j)$-entry equal to $\beta_{ij}-\alpha_{ij}$.
Let $d=s-{\rm rank}(N)$. 
{\color{black}
The \defword{stoichiometric subspace}, denoted by $S$,
is the image of $N$.}
A \defword{conservation-law matrix} of $G$, denoted by $W$, is a row-reduced $d\times s$-matrix whose rows form a basis of 
{\color{black} the orthogonal complement of $S$.} 
If there exists a choice of $W$ for which every entry is nonnegative and each column 
contains at least one nonzero entry (equivalently, each species occurs in at least one nonnegative conservation law), then $G$ is \defword{conservative}.

We denote the concentrations of the species $X_1,X_2, \ldots, X_s$ by $x_1, x_2, \ldots, x_s$, respectively. 
These concentrations, under the assumption of {\em mass-action kinetics}, evolve in time according to the following system of ODEs:
\begin{equation}\label{sys}
\dot{x}~=~f(x)~:=~N\cdot \begin{pmatrix}
\kappa_1 \, x_1^{\alpha_{11}} 
		x_2^{\alpha_{21}} 
		\cdots x_s^{\alpha_{s1}} \\
\kappa_2 \, x_1^{\alpha_{12}} 
		x_2^{\alpha_{22}} 
		\cdots x_s^{\alpha_{s2}} \\
		\vdots \\
\kappa_m \, x_1^{\alpha_{1m}} 
		x_2^{\alpha_{2m}} 
		\cdots x_s^{\alpha_{sm}} \\
\end{pmatrix}~,
\end{equation}
where $x=(x_1, x_2, \ldots, x_s)$, 
and each $\kappa_j \in \mathbb R_{>0}$ is called a \defword{reaction rate constant}. 
 By considering the rate constants as a vector of parameters $\kappa=(\kappa_1, \kappa_2, \dots, \kappa_m)$, we have polynomials $f_{\kappa,i} \in \mathbb Q[\kappa,x]$, for $i=1,2, \dots, s$.
For ease of notation, we often write $f_i$ rather than  $f_{\kappa,i}$.

A trajectory $x(t)$ beginning at a nonnegative vector $x(0)=x^0 \in
\mathbb{R}^s_{\geq 0}$ remains, for all positive time,
 in the following \defword{stoichiometric compatibility class} with respect to the \defword{total-constant vector} $c\coloneqq W x^0 \in {\mathbb R}^d$: 
\begin{align} \label{eqn:invtPoly}
\scc_c~\coloneqq~ \{x\in {\mathbb R}_{\geq 0}^s \mid Wx=c\}~.
\end{align}
A \defword{steady state} of~\eqref{sys} is a nonnegative concentration vector 
$x^* \in \mathbb{R}_{\geq 0}^s$ at which the 
right-hand sides of the 
ODEs~\eqref{sys}  vanish: $f(x^*) =0$.  
We distinguish between \defword{positive steady states} $x ^* \in \mathbb{R}^s_{> 0}$ and \defword{boundary steady states} 
$x^*\in {\mathbb R}_{\geq 0}^s\backslash {\mathbb R}_{>0}^s$. 
Also, a steady state $x^*$ is \defword{nondegenerate} if 
${\rm Im}\left( {\rm Jac}(f) (x^*)|_{S} \right)$ 
is the stoichiometric subspace $\St$.
(Here, ${\rm Jac}(f)(x^*)$ is the Jacobian matrix of $f$, with respect to $x$, at $x^*$.)  
A nondegenerate steady state is 
\defword{exponentially stable} if each of the $\sigma\coloneqq \dim(\St)$ nonzero eigenvalues of ${\rm Jac}(f)(x^*)$ has negative real part.

A network $G$ is \defword{multistationary} 
(respectively, \defword{bistable})
if, for some choice of positive rate-constant vector $\kappa \in \mathbb{R}^m_{>0}$, there exists a stoichiometric compatibility class~\eqref{eqn:invtPoly} 
with two or more positive steady states (respectively, exponentially stable positive steady states) of~\eqref{sys}. 
A network 
is \defword{monostationary}\footnote{Some authors define {\em monostationary} to be non-multistationary; 
the two definitions are equivalent for the ERK networks in this work.}
if, for every choice of positive rate constants, there is exactly one positive steady state in every stoichiometric compatibility class.

To analyze steady states within a stoichiometric compatibility class, we will 
use conservation laws 
in place of linearly dependent steady-state equations, as follows.
Let $I = \{i_1 < i_2< \dots < i_d\}$ denote the indices of the first nonzero coordinate of the rows of conservation-law matrix $W$.
Consider the function $\F: {\mathbb R}_{\geq 0}^s\rightarrow {\mathbb R}^s$ defined by 
\begin{equation}\label{consys}
f_{c,\kappa,i} =\F(x)_i :=
\begin{cases}
f_{i}(x)&~\text{if}~i\not\in I,\\
(Wx-c)_k &~\text{if}~i~=~i_k\in I .
\end{cases}
\end{equation}
We call system~\eqref{consys}, the system \defword{augmented by conservation laws}. By construction, positive roots of the system of polynomial equations $\F=0$ are precisely the positive steady states of~\eqref{sys} in the stoichiometric compatibility class~\eqref{eqn:invtPoly} defined by the total-constant vector $c$.

\subsection{Steady-state parametrizations} \label{sec:param}
Here we introduce steady-state parametrizations (Definition~\ref{def:parametrization-for-h}) 
and recall from~\cite{DPST} how to use them to determine whether a network is multistationary (Proposition~\ref{prop:c-general}).
Later we will see how to use parametrizations to detect Hopf bifurcations (Proposition~\ref{prop:hopfcriterion}).        

\begin{defn} \label{def:effective}
Let $G$ be a network with $m$ reactions and $s$ species, and let
$\dot x = f (x)$ denote the resulting mass-action system.  Denote by $W$ a
$d \times s$
row-reduced conservation-law matrix and by $I$ the set of indices of the first nonzero
coordinates of its rows. 
Enumerate the complement of $I$ as follows:
$[s] \setminus I = \{j_1 < j_2< \dots< j_{s-d}\}$.
A set of \defword{effective parameters}  
 for $G$ is formed by polynomials $\bar a_1(\kappa), \bar a_2(\kappa), \dots, \bar a_{\bar m}(\kappa) \in \mathbb{Q}(\kappa)$ for which the following hold: 
\begin{enumerate}[(i)]
	\item $\bar a_i(\kappa^*)$ is defined and, moreover,  
	$\bar a_i(\kappa^*)>0$
	for every $i=1,2,\dots,\bar m$ and for all $\kappa^* \in \mathbb{R}^m_{>0}$,
	\item the \defword{reparametrization map} below	is surjective:
\begin{align}\label{eq:ep}
	\bar a ~:~\mathbb{R}^m_{>0} & ~\to~ \mathbb{R}^{\bar m}_{>0} \\ \notag
				\kappa & ~\mapsto~ (\bar a_1(\kappa), \bar a_2(\kappa), \dots, \bar a_{\bar m} (\kappa))~,
\end{align} 
\item there exists an $(s-d)\times (s-d)$ matrix
	$M(\kappa)$ with entries in 
	$\mathbb{Q}(\kappa):=\mathbb{Q}(\kappa_1,\kappa_2,\dots, \kappa_m)$ such that:
	\begin{enumerate}[(a)]
	\item for all $\kappa^* \in \mathbb{R}^m_{>0}$, the matrix 
$M(\kappa^*)$ is defined 	and, moreover, 
	$\det M(\kappa^*)>0$, and 
	\item letting $(\bar h_{j_{\ell} })$ denote the functions  
	obtained from $(f_{j_\ell})$  
	as follows:
\begin{equation}\label{eq:linearchange}
(\bar h_{j_1}, \bar h_{j_2}, \dots, \bar h_{j_{s-d}})^{\top}
\quad := \quad M(\kappa) ~  (f_{j_1}, f_{j_2}, \dots, f_{j_{s-d}})^{\top}~,
\end{equation}
 every nonconstant coefficient in every  $\bar h_{j_{\ell}}$ is equal to   
  a rational-number multiple of some
 $\bar a_i(\kappa)$. 
	  \end{enumerate}
\end{enumerate}

Given such a set of effective parameters, we consider
for $\ell=1,2,\dots, s-d,$  polynomials
$h_{j_{\ell}}=h_{j_{\ell}} (a;x) \in \mathbb{Q}[a_1,a_2,\dots, a_{\bar m}][x]$ (here, the $a_i$'s are indeterminates)  such that:
\begin{align} \label{eq:effective-reln}
	\bar h_{j_{\ell}} ~=~ h_{j_{\ell}}|_{a_1 = \bar a_1(\kappa), ~\dots~,~ a_{\bar m} = \bar a_{\bar m}(\kappa)}~.
\end{align}
 For $i=1,2, \dots, s$ and any choice of $c \in \mathbb{R}^{d}_{>0}$
and $a \in \mathbb{R}^{\bar m}_{>0}$,
 set
\begin{equation}\label{consys-h}
\augmentH(x)_i~:=~
\begin{cases}
 h_i(a;x)
&~\text{if}~i  \notin I \\
(Wx-c)_k &~\text{if}~i=i_k\in  I.~
\end{cases}
\end{equation}
We call the function  $\augmentH: \mathbb{R}^s_{>0} \to \mathbb{R}^s $ an \defword{effective steady-state function} of $G$.
\end{defn}

The ``steady-state parametrizations'' that we will use in this work belong to a subclass of the ones introduced in~\cite{DPST}.  Thus, for simplicity, Definition~\ref{def:parametrization-for-h} 
below is more restrictive than~\cite[Definition~3.6]{DPST}. 
Specifically, our parametrizations have the form $\phi(\hat a; x)$, 
while those in~\cite{DPST} are of the form $\phi(\hat a; \hat x)$.


\begin{defn} \label{def:parametrization-for-h}
Let $G$ be a network with $m$ reactions, $s$ species, and
conservation-law matrix $W$.  Let $\F$ arise from $G$ and $W$ as in~\eqref{consys}.
Suppose that $\augmentH$ is an effective steady-state function of $G$, as in~\eqref{consys-h},
arising from a matrix $M(\kappa)$, as in~\eqref{eq:linearchange}, 
a reparametrization map $\bar a$, as in~\eqref{eq:ep}, 
and polynomials $h_{j_{\ell}}$'s as in~\eqref{eq:effective-reln}. 
The positive steady states of $G$ \defword{admit a positive parametrization with respect to $\augmentH$} 
 if there exists a function
 $	\phi : \mathbb{R}^{\hat m}_{>0} \times  \mathbb{R}^{s}_{>0} 
		\rightarrow
		 \mathbb{R}_{>0}^{\bar m}\times  \mathbb{R}_{>0}^{s}$, 
for some $\hat m \leq \bar m$, which we denote by 
		$(\hat a; x) \mapsto \phi(\hat a; x)$,
such that: 
	\begin{enumerate}[(i)]
	\item $\phi(\hat a;  x)$ extends the vector $(\hat a;  x)$. More precisely, 
	there exists a natural projection $\pi: \mathbb{R}^{\bar m}_{>0}\times
	  \mathbb{R}_{>0}^{s}\to \mathbb{R}^{\hat m}_{>0} \times \mathbb{R}^{ s}_{>0} $ 
	  such that $\pi \circ \phi$ is equal to the identity map.
	\item Consider any $(a;x) \in \mathbb R^{\bar m}_{>0} \times \mathbb  R_{>0}^s$.  Then, the equality 
		$h_i (a;x) = 0$ holds for every  $i \notin I$ 
		if and only if 
		 there exists
	 $\hat a^* \in  \mathbb{R}^{\hat m}_{>0} $ 
 such that $(a; x)=\phi(\hat a^*;x)$.
 \end{enumerate}
We call $\phi$ a \defword{positive parametrization} or a \defword{steady-state parametrization}.
 \end{defn}
 
\begin{defn}\label{def:Critical}
Under the notation and hypotheses of Definition~\ref{def:parametrization-for-h}, assume  
that the steady states of $G$ admit a positive parametrization with respect to $\augmentH$.
For such a positive parametrization $\phi$, the \defword{critical function} 
$C: \mathbb{R}^{\hat m}_{>0} \times  \mathbb{R}^{ s}_{>0}  \rightarrow   \mathbb{R}$ is given by:
\begin{equation*}
 C(\hat a;  x) \quad  = \quad \left(\det {\rm Jac}~\augmentH\right)|_{(a; x)=\phi(\hat a;  x)}~,
\end{equation*}
where ${\rm Jac}(\augmentH)$ denotes the Jacobian matrix of $\augmentH$ with respect to $x$. 
\end{defn}

The following result is a specialization\footnote{As noted earlier, here we consider parametrizations of the form $\phi(\hat a; x)$, while~\cite{DPST} allowed those of the form $\phi(\hat a; \hat x)$.  Also, ``conservative'' in Proposition~\ref{prop:c-general} can be generalized to ``dissipative''~\cite{DPST}.} of~\cite[Theorem 3.12]{DPST}:
\begin{proposition} \label{prop:c-general}
Under the notation and hypotheses of 
Definitions~\ref{def:effective}--\ref{def:Critical}, 
 assume also that $G$ is a conservative network without boundary steady states in any compatibility class.  Let $N$ denote the stoichiometric matrix of $G$.
\begin{enumerate}[(A)]
\item  {\bf Multistationarity.}
 $G$ is multistationary 
if there exists 
$(\hat a^*;  x^*) \in 
\mathbb{R}^{\hat m}_{>0} \times  \mathbb{R}^{ s}_{>0} $ 
such that 
\[
{\rm sign}( C(\hat a^*;  x^*) ) ~=~ (-1)^{\mathrm{rank}(N)+1}~.
\]
\item {\bf Monostationarity.}
$G$ is monostationary 
if for all 
$(\hat a;  x) \in \mathbb{R}^{\hat m}_{>0} \times  \mathbb{R}^{ s}_{>0} $, 
\[
{\rm sign}(C(\hat a;  x)) ~=~ (-1)^{\mathrm{rank}(N)}~.
\]
\end{enumerate} 
\end{proposition}

\subsection{Hopf bifurcations} \label{sec:hopf}
A \defword{simple Hopf bifurcation} is a bifurcation in which a single complex-conjugate pair of eigenvalues of the Jacobian matrix crosses the imaginary axis, while all other eigenvalues remain with negative real parts.  Such a bifurcation, if it is supercritical, generates nearby \defword{oscillations} or periodic orbits~\cite{liu}.

To detect simple Hopf bifurcations, we will use a criterion of Yang that characterizes Hopf bifurcations in terms of Hurwitz-matrix determinants (Proposition~\ref{prop:yang}). 
\begin{defn} \label{def:hurwitz}
The $i$-th \defword{Hurwitz matrix} of a univariate polynomial 
$p(\lambda)= b_0 \lambda^n + b_{1} \lambda^{n-1} + \cdots + b_n$  
is the following $i \times i$ matrix:
\[
H_i ~=~ 
\begin{pmatrix}
b_1 & b_0 & 0 & 0 & 0 & \cdots & 0 \\
b_3 & b_2 & b_1 & b_0 & 0 & \cdots & 0 \\
\vdots & \vdots & \vdots &\vdots & \vdots &  & \vdots \\
b_{2i-1} & b_{2i-2} & b_{2i-3} & b_{2i-4} &b_{2i-5} &\cdots & b_i
\end{pmatrix}~,
\]
in which the $(k,l)$-th entry is $b_{2k-l}$ 
  as long as $n \geq 2 k - l \geq 0$, and
  $0$ otherwise.
\end{defn}

Consider an ODE system parametrized by $\mu \in \mathbb{R}$:
    \begin{align*}
        \dot x ~=~ g_{\mu}(x)~,
    \end{align*}
where $x \in \mathbb{R}^n$, and $g_{\mu}(x)$ varies smoothly in $\mu$ and $x$.  Assume that $x^0 \in \mathbb{R}^n$ is a steady state of the system defined by $\mu_0$, that is, $g_{\mu_0}(x^0)=0$.  Assume, furthermore, that we have a smooth curve of steady states:
    \begin{align}  \label{eq:curve}
    \mu ~\mapsto ~ x(\mu)~
    \end{align}
(that is, $g_{\mu}\left( x(\mu) \right)= 0$ for all $\mu$) and that $x(\mu_0)=x^0$.  
Denote the characteristic polynomial of the Jacobian matrix of $g_{\mu}$, evaluated at $x(\mu)$, as follows:
    \begin{align*}
         p_{\mu}(\lambda) 
         ~:=~ \det \left( \lambda I - {\rm Jac}~ g_{\mu} \right)|_{x = x(\mu)}
         ~=~ \lambda^n + b_{1}(\mu) \lambda^{n-1} + \cdots + b_n(\mu)~, 
    \end{align*}
and, for $i=1,\dots, n$, define $H_i(\mu)$ to be the $i$-th Hurwitz matrix of $p_{\mu}(\lambda)$.


\begin{proposition} [Yang's 
criterion~\cite{yang-hopf}] \label{prop:yang}
Assume the above setup.  Then, there is a simple Hopf bifurcation at $x_0$ with respect to $\mu$
if and only if the following hold:
    \begin{enumerate}[(i)]
        \item $b_n(\mu_0)>0$,
        \item $\det H_1(\mu_0)>0$, $\det H_2(\mu_0)>0$, \dots, $\det H_{n-2}(\mu_0)>0$, and
        \item $\det H_{n-1}(\mu_0)= 0$ and $\frac{d( \det H_{n-1}(\mu) )}{d \mu}|_{\mu = \mu_0} \neq 0$. 
    \end{enumerate}

\end{proposition}

\subsection{Using parametrizations to detect Hopf bifurcations} \label{sec:param-Hopf}

Here we prove 
a new result on how to use steady-state parametrizations to detect Hopf bifurcations (Theorem~\ref{thm:hopf-criterion}).  The result, which uses Yang's criterion, 
is a straightforward generalization of the approach used in~\cite{mixed}.
We include it here to use later in Section~\ref{sec:results}, and we also expect it to be useful in future work.

\begin{lemma} \label{lem:reduced-char-poly-general}
Let $G$ be a network with $s$ species, $m$ reactions, and $d$ conservation laws.  
Denote the ODEs by $\dot x = f(x)$, as in~\eqref{sys}.
Assume that the positive steady states of $G$ admit a positive parametrization $\phi$ with respect to an effective steady-state function for which the reparametrization map~\eqref{eq:ep} is just the identity map.  In other words, the effective parameters $\bar{a}_i$ are the original rate constants $\kappa_i$, and so we write $\phi: \mathbb{R}^{\hat m}_{>0} \times  \mathbb{R}^{s}_{>0} 
		\rightarrow
		 \mathbb{R}_{>0}^{ m}\times  \mathbb{R}_{>0}^{s}$ as 
		 $(\hat{\kappa}; x) \mapsto \phi(\hat{\kappa}; x)$.
Assume moreover that each coordinate of $\phi_i$ is a rational function: $\phi_i(\hat{\kappa}; x) \in \mathbb{Q}(\hat{\kappa}; x)$ for $i=1,2,\dots, \hat{m}+s$.
Then the following is a univariate, degree-$(s-d)$ polynomial in $\lambda$, with coefficients in $\mathbb{Q}(\hat{\kappa}; x)$:
\begin{align} \label{eq:reduced-char-poly-general}
    q(\lambda)~:=~  \frac{1}{\lambda^{d}} ~ \det \left( \lambda I - {\rm Jac}~f \right) |_{(\kappa; x)~=~ \phi(\hat{\kappa}; x)} ~.
\end{align}
\end{lemma}

\begin{proof}
This result is straightforward from the fact that the characteristic polynomial of ${\rm Jac}(f)$ is a polynomial of degree $s$ and has zero as a root with multiplicity $d$ (because of the $d$ conservation laws). \end{proof}

\begin{theorem}[Hopf-bifurcation criterion] \label{thm:hopf-criterion}
Assume the hypotheses of Lemma~\ref{lem:reduced-char-poly-general}.
Let $\mathfrak{h}_i$ (for $i=1,2,\dots, s-d$)
be the determinant of the $i$-th Hurwitz matrix of 
$q(\lambda)$ in~\eqref{eq:reduced-char-poly-general}.  Let $\kappa_j$ be one of the rate constants in the vector $\hat{\kappa}$.  Then the following are equivalent:
\begin{enumerate}[(1)]
    \item there exists a rate-constant vector $\kappa^* \in \mathbb{R}^{m}_{>0}$ such that the resulting system~\eqref{sys} exhibits a simple Hopf bifurcation with respect to $\kappa_j$ at some 
    $x^* \in \mathbb{R}^{s}_{>0}$, and
    \item there exist 
        $\hat{\kappa}^*\in \mathbb{R}^{\hat m}_{>0}$ 
        and 
        $x^* \in \mathbb{R}^{s}_{>0}$ 
        such that 
        \begin{enumerate}[(i)]
        \item the constant term of the polynomial $q(\lambda)$, when evaluated at 
            $(\hat{\kappa}; x)=(\hat{\kappa}^*; x^*)$, is positive, 
        \item $\mathfrak{h}_1 (\hat{\kappa}^*; x^*) >0$, 
            $\mathfrak{h}_2 (\hat{\kappa}^*; x^*) >0$, \dots , 
            $\mathfrak{h}_{s-d-2} (\hat{\kappa}^*; x^*) >0$~, and
        \item $\mathfrak{h}_{s-d-1} (\hat{\kappa}^*; x^*) =0$ and 
        $\frac{ \partial \mathfrak{h}_{s-d-1} }{\partial \kappa_j } \vert_{(\hat{\kappa}; x)=(\hat{\kappa}^*; x^*)} \neq 0$.
        \end{enumerate}
    \end{enumerate}
Moreover, given $\hat{\kappa}^*$ and $x^*$ as in (2), 
a simple Hopf bifurcation with respect to $\kappa_j$ occurs at $x^*$ when the 
vector of rate constants is taken to be $\kappa^*:=\widetilde{\pi} (\phi(\hat{\kappa}^*; x^*))$.  Here, $\widetilde{\pi}:  \mathbb{R}_{>0}^{m}\times  \mathbb{R}_{>0}^{s} \to \mathbb{R}_{>0}^m$ is the natural projection. 
\end{theorem}

\begin{proof}
Due to the $d$ conservation laws, 
we apply Yang's criterion (Proposition~\ref{prop:yang})
to:
\[
    \frac{1}{\lambda^d} \det (\lambda I - {\rm Jac}~f ) \vert_{x=x^*,~\kappa_i=\kappa^*_i {\mathrm ~for~all~} i \neq j}~.
\]
Now our result follows directly from Proposition~\ref{prop:yang} and Definition~\ref{def:parametrization-for-h}.
\end{proof}

\begin{remark} \label{rmk:hopf}
Theorem~\ref{thm:hopf-criterion} easily generalizes 
beyond parametrizations of the form $\phi(\hat{\kappa}; {x})$
to those of the form $\phi(\hat{\kappa}; \hat{x})$ or  $\phi(\kappa; \hat{x})$.  
Indeed, one of the form  $\phi(\kappa; \hat{x})$
was used
 in~\cite{mixed} 
to establish Hopf bifurcations in a mixed-mechanism phosphorylation system.
\end{remark}

\section{ERK networks and steady-state parametrizations} \label{sec:parametrizations}
Here we introduce steady-state parametrizations for the full ERK network and also irreducible and reduced versions of the network (Propositions~\ref{prop:param-irrev} and~\ref{prop:param-reduced}).  

\subsection{The (full) ERK network}

\begin{table}[hbt]
  \centering
    \begin{tabular}[ht]{|cccccccccccc|} \hline
      $x_1$ & $x_2$ & $x_3$ & $x_4$ & $x_5$ & $x_6$ & $x_7$ & $x_8$ & $x_9$ & $x_{10}$ &
      $x_{11}$ & $x_{12}$
      \\ \hline
$S_{00}$ &  $E$ & $F$ & $S_{11} F$ & $S_{10} F$ & $S_{01} F$ & $S_{01} E$ & $S_{10}E$ 
	& $S_{01}$ & $S_{10}$  & $S_{00} E$  & $S_{11}$ \\
\hline
    \end{tabular}
  \caption{
      Assignment of variables to species for the ERK network~in Figure~\ref{fig:erk_sys}.   }
  \label{tab:variables}
\end{table}

For the full ERK network shown earlier in Figure \ref{fig:erk_sys}, 
we let $x_1,x_2,\ldots,x_{12}$ denote the concentrations of the species  
in the order given in Table~\ref{tab:variables}.
The resulting ODE system~\eqref{sys} is as follows: 
\begin{align} \label{eq:ODE-full}
\begin{split}
    \dot{x_1} & = -k_1x_1x_2+k_2 x_{11}+\lcat x_5+n_3 x_6 \\ 
     \dot{x_2} &  = -k_1 x_1 x_2-\kon x_2 x_9-m_2 x_2 x_{10}+k_2 x_{11}+\kcat x_7+\koff x_7+m_1 x_8+m_3 x_8 \\ 
    \dot{x_3} &  = -\ell_1 x_3 x_{12}-\lon x_3 x_{10}-n_1 x_3 x_9+\ell_2 x_4+\lcat x_5+\loff x_5+n_2 x_6+n_3 x_6 \\ 
    \dot{x_4} & = \ell_1 x_3 x_{12} -\ell_2 x_4-\ell_3 x_4 \\ 
    \dot{x_5} & = \lon x_3  x_{10}+\ell_3 x_4-\lcat x_5-\loff x_5 \\ 
    \dot{x_6} & = n_1 x_3 x_9-n_2 x_6-n_3 x_6 \\ 
    \dot{x_7} & = \kon x_2 x_9+k_3 x_{11}-\kcat x_7-\koff x_7 \\ 
    \dot{x_8} & = m_2 x_2 x_{10}-m_1 x_8-m_3 x_8 \\ 
    \dot{x_9} & = -\kon x_2 x_9-n_1 x_3 x_9+\koff x_7+n_2 x_6 \\ 
    \dot{x_{10}} & = -\lon x_3 x_{10}-m_2 x_2 x_{10}+\loff x_5+m_1 x_8 \\ 
    \dot{x_{11}} & = k_1x_1x_2-k_2x_{11}-k_3x_{11}\\ 
    \dot{x_{12}} & = -\ell_1 x_3 x_{12} +\kcat x_7+\ell_2 x_4+m_3 x_8 \\ 
    \end{split}
\end{align}


There are 18 rate constants $k_i,\ell_i,m_i,n_i$. 
The 3 conservation laws correspond to the total amounts of substrate $S$, kinase $E$, and phosphatase $F$, respectively:

\begin{align}\label{eq:conservation}
\begin{split}
x_1 + x_4 + x_5 + x_6 + x_7 + x_8 + x_9 + x_{10} + x_{11} + x_{12} ~&=~ \stot ~=:~ c_1\\
x_2 + x_7 + x_8 + x_{11} ~&=~ \etot ~=:~ c_2\\
x_3 + x_4 + x_5 + x_6  ~&=~ \ftot~=:~ c_3 .
\end{split}
\end{align}

A steady-state parametrization for the full ERK network was given in~\cite[Examples 3.1 and 3.7]{DPST}.  That parametrization, however, can not specialize to accommodate irreversible versions of the network (in the effective parameters given in~\cite{DPST}, two of the denominators are $\kon$ and $\lon$, so we can not set those rate constants to 0).
So, in the next subsection, we give an alternate steady-state parametrization that, although quite similar to the one in~\cite{DPST}, specializes when considering irreversible versions of the network (see Proposition~\ref{prop:param-irrev}).

\subsection{Irreversible versions of the ERK network} \label{sec:irreversible}
Here we consider networks obtained from the full ERK network (Figure~\ref{fig:erk_sys}) by making some  reversible reactions irreversible.
Specifically, we delete one or more of the reactions marked in blue in Figure~\ref{fig:erk_sys}.
Our motivation for removing those specific reactions (the ones with rate constants $k_2, \kon, m_1, \ell_2, \lon, n_2$) rather than any of their opposite reactions 
is to preserve the main reaction pathways (from $S_{00}$ to $S_{11}$, as well as $S_{10}$ to $S_{11}$, 
$S_{11}$ to $S_{00}$, 
and
$S_{01}$ to $S_{00}$).  At the same time, we do not remove the reactions for $\koff$ or $\loff$, so that we can still pursue Question~\ref{q:main} (which involves $\koff$ and $\loff$) in a model with fewer reactions. We instead allow the removal of reactions $\kon$ and $\lon$.


\begin{proposition}[Steady-state parametrization for full and irreversible ERK networks]  \label{prop:param-irrev}
Let $\mathcal{N}$ be the full ERK network or any 
network obtained from the full ERK network 
by deleting one or more the reactions corresponding to rate constants $k_2, \kon, m_1, \ell_2, \lon, n_2$ (marked in blue in Figure~\ref{fig:erk_sys} ).
Let $\mathds{1}_{k_2}$ denote the indicator function that is 1 if the reaction 
labeled by $k_2$ is in $\mathcal N$ and 0 otherwise; analogously, we also define 
$\mathds{1}_{\kon}$, 
$\mathds{1}_{m_1}$, 
$\mathds{1}_{\ell_2}$, 
$\mathds{1}_{\lon}$, and 
$\mathds{1}_{n_2}$.  
Then $\mathcal N$ admits an effective steady-state function $h_{c,a}: \mathbb{R}^{12}_{>0} \to \mathbb{R}^{12}$ given by: 
\begin{align}\label{eq:effectivefunction_irrev}
\notag
h_{c, a,1} ~&=~ x_{1}+x_4+x_5+x_6+x_7+x_8 +x_9+x_{10}+x_{11}+x_{12} -c_1~,  \\ 
\notag
h_{c, a,2} ~&=~ x_2 +x_7+x_8+ x_{11} -c_2~,  \\
\notag
h_{c, a,3} ~&=~x_3+ x_4+x_5+x_6 -c_3~, \\ 
\notag
h_{c, a,4}  ~&=~  a_{12}x_3x_{12}-x_4~,  \\
\notag
h_{c, a,5}  ~&=~  a_3x_4 - x_5 -a_2x_8 ~, \\ 
h_{c, a,6}  ~&=~  a_{13}x_3x_{9}-x_6~,  \\
\notag
h_{c, a,7}  ~&=~  a_5x_{11} -a_4x_6- x_7 ~,\\ 
\notag
h_{c, a,8}  ~&=~  a_{11}x_2x_{10}-x_8~,\\
\notag
h_{c, a,9} ~&=~  a_9x_7-  \mathds{1}_{\kon} a_8x_2x_9 - x_6 ~,\\ 
\notag
h_{c, a,10}  ~&=~  a_7x_5-  \mathds{1}_{\lon} a_6x_3x_{10} - x_8~,\\
\notag
h_{c, a,11}  ~&=~ a_{10}x_1x_{2}-x_{11}~,\\
\notag
h_{c, a,12}  ~&=~ x_7 -a_1x_5~.
\end{align}
Moreover, 
with respect to this effective steady-state function, 
the positive steady states of $\mathcal N$ 
admit the following positive parametrization: 
  \begin{align*} 
    \phi : \mathbb{R}^{2+\mathds{1}_{\kon}+\mathds{1}_{\lon} + 12 }_{>0} ~& \to~ \mathbb{R}^{13+12}_{>0} \\
    (\hat{a};~x_1, x_2,\dots, x_{12} ) ~&\mapsto ~({a}_1,{a}_2,\dots,{a}_{13},~x_1, x_2,\dots, x_{12})~, \notag
  \end{align*}  
  given by
\begin{align} \label{eq:param-irrev-details} \notag
a_1 ~&:=~ \dfrac{x_7}{x_5} & \quad 
a_3 ~&:=~ \dfrac{a_2x_8+x_5}{x_4} & \quad 
a_5 ~&:=~ \dfrac{a_4x_6+x_7}{x_{11}}  
\\ 
a_7 ~&:=~ \dfrac{ \mathds{1}_{\lon} a_6x_3x_{10}+x_8}{x_5} & \quad
a_9 ~&:=~ \dfrac{ \mathds{1}_{\kon} a_8x_2x_9+x_6}{x_7}& \quad 
a_{10} ~&:=~ \dfrac{x_{11}}{x_1x_2} 
\\ \notag
a_{11} ~&:=~ \dfrac{x_8}{x_2x_{10}} & \quad
a_{12} ~&:=~ \dfrac{x_4}{x_3x_{12}} & \quad 
a_{13} ~&:=~ \dfrac{x_6}{x_3x_9} ~.
\end{align}
Here, $\hat a = (a_2, a_4, a_6, a_8)$ if $\mathcal N$ contains the reactions labeled by $\kon$ and $\lon$, 
and
$\hat a = (a_2, a_4, a_6)$ if $\mathcal N$ contains the reaction $\lon$ but not $\kon$, and so on.  
\end{proposition}

\begin{proof}
We will show that the map 
$\bar a:  \mathbb{R}^{12+\mathds{1}_{k_2} + 
\mathds{1}_{\kon} +
\mathds{1}_{m_1} + 
\mathds{1}_{\ell_2} + 
\mathds{1}_{\lon} + 
\mathds{1}_{n_2} }_{>0} \to \mathbb{R}^{11+  \mathds{1}_{\kon} + \mathds{1}_{\lon} }_{>0}$,
defined as follows, is a reparametrization map as in~\eqref{eq:ep}:
\begin{equation}\label{eq:effective_irrev}
\begin{tabular}{lllll}
$\bar a_1=\frac{\ell_{cat}}{\kcat}$,
& $\bar a_2=\frac{m_{3}}{\lcat}$,& $\bar a_3=\frac{\ell_3}{\lcat}$,
& $\bar a_4=\frac{n_3}{\kcat}$, & $\bar a_5=\frac{k_3}{\kcat}$, \\
 $\bar a_6=\frac{ \mathds{1}_{\lon} \lon }{m_3}$, & $\bar a_7=\frac{ \loff }{m_3}$,&
$\bar a_8=\frac{ \mathds{1}_{\kon} \kon }{n_3}$,& $\bar a_9=\frac{ \koff }{n_3}$,&$\bar a_{10}=\frac{k_1}{ \mathds{1}_{k_2} k_2+ k_3}$, \\ 
$\bar a_{11}=\frac{m_2}{ \mathds{1}_{m_1} m_1+m_3}$, & $\bar a_{12}=\frac{\ell_1}{ \mathds{1}_{\ell_2} \ell_2+\ell_3}$, &
$\bar a_{13}=\frac{n_1}{ \mathds{1}_{n_2} n_2+n_3}$. & \\
\end{tabular}
\end{equation}
In particular, we remove the effective parameter $\bar{a}_6$ (respectively, $\bar{a}_8$) if 
$\mathds{1}_{\lon}=0$ (respectively, $\mathds{1}_{\kon}=0$).  
Notice that each $\bar a_i$ (if it is not removed) is defined and positive for all $\kappa = (k_1,\dots, n_3) \in \mathbb{R}^{12+\mathds{1}_{k_2} + 
\mathds{1}_{\kon} +
\mathds{1}_{m_1} + 
\mathds{1}_{\ell_2} + 
\mathds{1}_{\lon} + 
\mathds{1}_{n_2} }_{>0}$.  

We  must show that the map $\bar a$
is surjective.  
Indeed,
given $a \in \mathbb{R}^{11+  \mathds{1}_{\kon} + \mathds{1}_{\lon} }_{>0}$,
it is easy to check that $a$ is the image under $\bar a$ of the vector obtained by removing 
every 0 coordinate from the following vector: 
{\small
\begin{align*}
    &(k_1,
    k_2, k_3, \kcat, \kon, \koff, 
    {\ell}_1,
    {\ell}_2, {\ell}_3, \lcat, \lon, \loff, 
    m_1, m_2, m_3,
    n_1, n_2, n_3
    ) 
    ~=~ \\
    & ( 
    ( \mathds{1}_{k_2} +a_5)a_{10}, 
     \mathds{1}_{k_2} , 
    a_5, 
    1, 
     \mathds{1}_{\kon} a_4 a_8, 
    a_4 a_9, 
    ( \mathds{1}_{\ell_2} +a_1 a_3)a_{12}, 
     \mathds{1}_{\ell_2} , 
    a_1 a_3, 
    a_1, 
     \mathds{1}_{\lon} a_1 a_2 a_6, 
    a_1 a_2 a_7, 
     \mathds{1}_{m_1} ,  \\ 
    & \quad \quad ( \mathds{1}_{m_1} +a_1 a_2)a_{11}, 
    a_1 a_2, 
    ( \mathds{1}_{n_2} +a_4)a_{13}, 
     \mathds{1}_{n_2} , 
    a_4 
    )~.
\end{align*}
}

Next, consider the following $9 \times 9$ matrix:
{\footnotesize
\begin{align} \label{eq:M-kappa-irrev}
M(\kappa)=
\left(\begin{array}{cccccccccccc}
\frac{1}{\mathds{1}_{\ell_2} \ell_2+\ell_3} &  0& 0& 0  & 0&0 & 0&0 & 0\\
0 & \frac{1}{\lcat }  & 0 & 0   & \frac{1}{\lcat } & 0& \frac{1}{ \lcat }&0&0\\
 0 & 0 & \frac{1}{\mathds{1}_{n_2} n_2+n_3}& 0 & 0 & 0 & 0&0&0\\
 0 &  0 & \frac{1}{ \kcat }    &  \frac{1}{\kcat }& 0 & \frac{1}{\kcat }& 0&0&0\\
  0 &  0 & 0     &  0& \frac{1}{\mathds{1}_{m_1} m_1+m_3} & 0& 0&0&0\\
  0     &  0 &  \frac{1}{n_{3}}   & 0   & 0       &   \frac{1}{n_{3}}  & 0& 0 &0 \\
   0&  0& 0& 0  & \frac{1}{m_{3}}& 0 &\frac{1}{m_{3}}&0 & 0\\
    0&  0& 0& 0  & 0& 0 &0 &\frac{1}{\mathds{1}_{k_2} k_2+ k_3}& 0\\
     \frac{1}{ \kcat }&  \frac{1}{ \kcat }& 0& 0 & \frac{1}{\kcat }& 0 &\frac{1}{\kcat }&0 & \frac{1}{\kcat }
\end{array}
\right)~.
\end{align}
} \noindent
It is straightforward to check that 
 ${\rm det} M(\kappa)$ 
is the product of all diagonal terms, and hence is positive for all $\kappa\in \mathbb{R}^{ 12+\mathds{1}_{k_2} + 
\mathds{1}_{\kon} +
\mathds{1}_{m_1} + 
\mathds{1}_{\ell_2} + 
\mathds{1}_{\lon} + 
\mathds{1}_{n_2}}_{>0}$. 

The mass-action ODEs of $\mathcal{N}$ are obtained from those~\eqref{eq:ODE-full} of the full ERK network by replacing the rate constants $k_2, \kon, m_1, \ell_2, \lon, n_2$, respectively, by 
$\mathds{1}_{k_2} k_2$, 
$\mathds{1}_{\kon} \kon$, 
$\mathds{1}_{m_1} m_1$, 
$\mathds{1}_{\ell_2} \ell_2$, 
$\mathds{1}_{\lon} \lon$, and 
$\mathds{1}_{n_2} n_2$, respectively. 
To the right-hand sides of these ODEs, we apply the recipe given in equations \eqref{eq:linearchange}--\eqref{consys-h}, using 
the effective parameters $\bar a_i$ in~\eqref{eq:effective_irrev},
the matrix $M(\kappa)$ in~\eqref{eq:M-kappa-irrev}, and
the conservation-law matrix $W$ arising from the conservation laws~\eqref{eq:conservation}.
It is straightforward to check that the result is the function 
$\augmentH(x)$ 
given in~\eqref{eq:effectivefunction_irrev}.

Observe that, for the non-conservation-law equations 
$h_{c, a,4}, \dots, h_{c, a,12}$
in~\eqref{eq:effectivefunction_irrev},
each non-constant coefficient is, up to sign, one of the $a_i$'s.
Hence, the $\bar a_i$'s in~\eqref{eq:effective_irrev} are effective parameters, and the function
in~\eqref{eq:effectivefunction_irrev} is an effective steady-state function. 
Finally, the fact that $\phi$ is a positive parametrization with respect to~\eqref{eq:effectivefunction_irrev} (as in Definition~\ref{def:parametrization-for-h}) follows directly from comparing equations~\eqref{eq:effectivefunction_irrev}
 and~\eqref{eq:param-irrev-details}.
\end{proof}

\begin{remark}[Multistationarity depends on only $\kon$ and $\lon$] \label{rmk:only-depends-k-on-l-on}
Proposition~\ref{prop:param-irrev} considers any network obtained by deleting any (or none) of the six reactions labeled by $k_2$, $\kon$, $m_1$, $\ell_2$, $\lon$, $n_2$.  Nonetheless, the resulting steady-state parametrization~\eqref{eq:param-irrev-details} depends on $\kon$ and $\lon$ but not any of the other rate constants. Thus, 
multistationarity for these irreversible networks depends only on whether the network contains $\kon$ and $\lon$ (see Theorem~\ref{thm:bistable}).
\end{remark}

\subsection{The reduced ERK network} \label{sec:reduced}
In the previous subsection, we consider irreversible versions of the ERK network.  Now we further reduce the network by additionally removing some ``intermediate complexes'' (namely, $S_{10}E$ and $S_{01}F$).
These operations yield the reduced ERK network in Figure~\ref{fig:red_erk_net}.
Note that in the process of removing intermediates, the reactions $m_2$ and $m_3$ (similarly, $n_1$ and $n_3$) are collapsed into a single reaction labeled $m$ (respectively, $n$). 
A biological motivation for collapsing these reactions is the fact that intermediates are usually short-lived, so the simpler model may approximate the dynamics well.  

\begin{figure}[ht]
    \centering
    \includegraphics[scale=1.2]{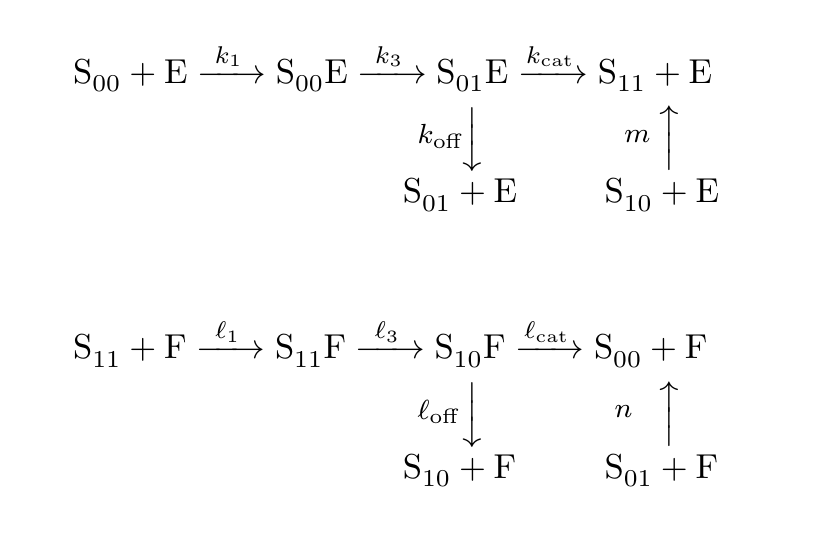}
    \caption{\defword{Reduced ERK network}.}
    \label{fig:red_erk_net}
\end{figure}

Our notion of removing intermediates matches that of Feliu and Wiuf~\cite{FeliuWiuf}, who initiated the recent interest in the question of when dynamical properties are preserved when intermediates are added or removed (e.g., $S_{10}+E \leftrightarrows S_{10}E \to S_{11} + E$ versus $S_{10} \to S_{11} $).
Our work, therefore, fits into this circle of ideas~\cite{approx-elim, Freitas1, core}. 

\begin{table}[hbt]
  \centering
    \begin{tabular}[ht]{|cccccccccc|} \hline
      $x_1$ & $x_2$ & $x_3$ & $x_4$ & $x_5$ & $x_6$ & $x_7$ & $x_8$ & $x_9$ & $x_{10}$ \\ \hline
$S_{00}$ &  $E$ & $S_{00} E$ & $S_{01} E$ & $S_{11}$ & $S_{01}$ & $S_{10}$ & $F$ & $S_{11} F$ & $S_{10} F$\\
\hline
    \end{tabular}
  \caption{
      Assignment of variables to species for the reduced ERK network~in Figure~\ref{fig:red_erk_net}.
      (Many of the variables that are also in the full ERK, in Table~\ref{tab:variables}, have been relabeled.)
  }
  \label{tab:variables-reduced}
\end{table}

In the reduced ERK network, the remaining 10 rate constants are as follows:
$k_1, k_3, \kcat, \koff, m, \ell_1, \ell_3, \lcat, \loff, n$.
Letting $x_1,x_2, \ldots, x_{10}$ denote the species
    concentrations in the order given in
    Table~\ref{tab:variables-reduced}, 
the resulting mass-action kinetics ODEs are as follows:
\begin{align*} 
    \notag
\dot{x_1} ~&=~ -k_1 x_{1} x_{2}+n x_{6} x_{8}+\lcat x_{10} & &~=:~ f_1 \\ 
    \notag
\dot{x_2} ~&=~ -k_1 x_{1} x_{2}+\kcat x_{4}+\koff x_{4} & &~=:~ f_2 \\ 
    \notag
\dot{x_3} ~&=~ k_1 x_{1} x_{2}-k_3 x_{3} & &~=:~ f_3
\\ 
    \notag
\dot{x_4} ~&=~ k_3 x_{3}-\kcat x_{4}-\koff x_{4} & &~=:~ f_4 
\end{align*}
\begin{align} \label{eq:ODE-reduced}
\dot{x_5} ~&=~ m x_{2} x_{7}-\ell_1 x_{5} x_{8}+\kcat x_{4} & &~=:~ f_5 
\\ 
    \notag
\dot{x_6} ~&=~ -n x_{6} x_{8}+\koff x_{4} & &~=:~ f_6 \\ 
    \notag
\dot{x_7} ~&=~ -m x_{2} x_{7}+\loff x_{10} & &~=:~ f_7 \\ 
    \notag
\dot{x_8} ~&=~ -\ell_1 x_{5} x_{8}+\loff x_{10}+\lcat x_{10} & &~=:~ f_8 \\ 
    \notag
\dot{x_9} ~&=~ \ell_1 x_{5} x_{8}-\ell_3 x_{9} & &~=:~ f_9 \\ 
    \notag
\dot{x_{10}} ~&=~ -\loff x_{10}+\ell_3 x_{9}-\lcat x_{10} & &~=:~ f_{10}.
\end{align}

The 3 conservation equations are:
\begin{align} \label{eq:cons-law-reduced}
\notag
x_1+x_3+x_4+x_5+x_6+x_7+x_9+x_{10} ~&=~ S_\text{{tot}} ~=:~ c_1\\
x_2+x_3+x_4 ~&=~ E_{\text{tot}} ~=:~ c_2\\
\notag
x_8+x_9+x_{10} ~&=~ F_{\text{tot}}~=:~ c_3.
\end{align}

\begin{proposition}[Steady-state parametrization for reduced ERK network]  \label{prop:param-reduced}
The reduced ERK network (Figure~\ref{fig:red_erk_net}) 
admits an effective steady-state function 
$h_{c,a}: \mathbb{R}^{10}_{>0} \to \mathbb{R}^{10}$ given by: 

\begin{align}\label{eq:effectivefunction_reduced}
\begin{tabular}{ll}
$h_{c, a,1} = x_1+x_3+x_4+x_5+x_6+x_7+x_9+x_{10} - c_1,$  & 
$h_{c, a,2} = x_2+x_3+x_4 - c_2,$  \\
$h_{c, a,3} = -(\kcat+\koff)\lcat x_{10} + k_1 \kcat x_1 x_2 ,$ & 
$h_{c, a,4}  =              k_3x_3 - (\kcat+\koff)x_4$,  \\
$h_{c, a,5}  =               \loff x_{10}-m x_2x_7 $, & 
$h_{c, a,6}  =              \ell_1 x_5x_8-(\lcat+\loff)x_{10} $,  \\
$h_{c, a,7}  =              \ell_3 x_9 - (\lcat+\loff)x_{10}$,& 
$h_{c, a,8}  = x_8+x_9+x_{10} - c_3$,\\
$h_{c, a,9} =                \kcat x_4 - \lcat x_{10}$,& 
$h_{c, a,10}  =      \koff \lcat x_{10} - \kcat n x_6 x_8$.\\        
\end{tabular}
\end{align}
Moreover, 
with respect to this effective steady-state function, 
the positive steady states admit the following positive parametrization:
  \begin{align} \label{eq:param-reduced}
    \phi : \mathbb{R}^{3+10}_{>0} ~& \to~ \mathbb{R}^{10+10}_{>0} \\
    (\kcat, \koff, \loff,~x_1, x_2,\dots, x_{10}) ~&\mapsto ~({\kappa}_1,\kappa_3,\kcat,\koff, m ,\ell_1,\ell_3,\lambda_{\rm cat},\loff, n ,~x_1, x_2,\dots, x_{10})~, \notag
  \end{align}  
  given by
    \begin{align} \label{eq:param-reduced-details} \notag
k_1 ~&:=~ \dfrac{(\kcat+\koff) x_4}{x_1x_2} & \quad 
k_3 ~&:=~ \dfrac{(\kcat+\koff)x_4}{x_3} & \quad 
m ~&:=~ \dfrac{\loff x_{10}}{x_2x_7} & \quad 
\ell_1 ~&:=~ \dfrac{\loff x_{10}+\kcat x_4}{x_5x_8}\\ 
\ell_3 ~&:=~ \dfrac{\loff x_{10}+\kcat x_4}{x_9}& \quad 
\lcat ~&:=~ \dfrac{\kcat x_4}{x_{10}} & \quad 
n ~&:=~ \dfrac{\koff x_4}{x_6x_8}.
    \end{align}
In particular, the image of $\phi$ is the following set of pairs of positive steady states and rate constants:

\[  
    \{ ( k^*; x^*) \in \mathbb{R}^{10+10}_{>0} \mid x^* \text{ is a steady state of \eqref{eq:ODE-reduced} when } k=k^*\}~.
\]
Here, $k$ denotes the vector $(k_1, k_3, \kcat, \koff, m, \ell_1, \ell_3, \lcat, \loff, n)$.
\end{proposition}

\begin{proof}
Let $W$ denote the conservation-law matrix arising from the conservation laws~\eqref{eq:cons-law-reduced} for the reduced ERK network.  Then $I=\{1,2,8\}$ is the set of indices of the first nonzero coordinates of the rows of $W$.
We take $\mathbb{Q}(\kcat, \koff)$-linear combinations of the $f_i$'s in~\eqref{eq:ODE-reduced}, where $i \notin I$,
to obtain the following binomials in the $x_i$'s:
\begin{align*} 
    h_3 ~&:=~ (\kcat+\koff)(f_5+f_7+f_9+f_{10})+\kcat(f_3+f_4)
            &
            ~&=~ -(\kcat+\koff)\lcat x_{10} + k_1 \kcat x_1 x_2 \\
    h_4 ~&:=~ f_4
            &
            ~&=~ k_3x_3 - (\kcat+\koff)x_4\\
    h_5 ~&:=~ f_7
            &
            ~&=~ \loff x_{10}-m x_2x_7 \\
    h_6 ~&:=~ f_9+f_{10}
            &
            ~&=~ \ell_1 x_5x_8-(\lcat+\loff)x_{10} \\
    h_7 ~&:=~ f_{10}
            &
            ~&=~ \ell_3 x_9 - (\lcat+\loff)x_{10} \\
    h_{9} ~&:=~ f_5+f_7+f_9+f_{10}
            &
            ~&=~ \kcat x_4 - \lcat x_{10}\\
    h_{10} ~&:=~ \kcat f_6 - \koff(f_5+f_7+f_9+f_{10}) 
            &
            ~&=~  \koff \lcat x_{10} - \kcat n x_6 x_8. 
\end{align*}

Consider the (above) linear transformation from $f_i$ to $h_i$ $(i\not\in I)$.
Let $M$ denote the corresponding matrix representation ($M$ plays the role of the matrix denoted by $M(\kappa)$ in Definition~\ref{def:effective}).
It is straightforward to check that $\det M = \kcat^2$, which is positive when $\kcat>0$. 

Consider the reparametrization map $\bar a:\mathbb{R}^{10} \to \mathbb{R}^{10}$ defined by the identity map (and so is surjective).  Then $\bar a$, together with the conservation-law matrix $W$ and 
the matrix $M$, 
yield (as in Definition~\ref{def:effective}\footnote{In this case, Definition~\ref{def:effective}(iii)(b) 
requires every nonconstant coefficient in the effective steady-state function~\eqref{eq:effectivefunction_reduced} to be a 
rational-number multiple of one of the rate constants.  However, for the non-conservation-law equations in~\eqref{eq:effectivefunction_reduced},
many of the non-constant coefficients -- such as $\koff \lcat$ -- are not rational-number multiples of one of the rate constants.  Nonetheless, these coefficients are all polynomials in the rate constants, and the relevant results in~\cite{DPST} hold in that generality.})
the effective steady-state function 
$\augmentH(x)$ 
given in~\eqref{eq:effectivefunction_reduced}.

To show that $\phi$ is a positive steady-state parametrization with respect to~\eqref{eq:effectivefunction_reduced},
as in Definition~\ref{def:parametrization-for-h}, 
it suffices to show the following claim:

\noindent
{\bf Claim:} For every $(k^*; x^*) \in \mathbb{R}^{10+10}_{>0}$, the steady-state condition holds -- namely, $h_i( k^*; x^*)=0$ for all $i \in \{3,4,5,6,7,9,10 \}$ -- if and only if $\phi (\kcat^*, \koff^*, \loff^*; x^*) = (k^*;x^*)$.

For the ``$\Rightarrow$'' direction, assume $h_i(k^*;x^*)=0$ for all $i$.  
Then $h_9(k^*;x^*)=0$ implies that 
\begin{align} \label{eq:k_9}
    \lcat^* ~=~ \dfrac{\kcat^* x^*_4}{x^*_{10}}~.
\end{align}
In other words, $\lambda_{\rm cat} $ -- when evaluated at 
$(\kcat, \koff, \loff;x) =(\kcat^*, \koff^*, \loff^*;x^*) $
-- equals $\lcat^*$.
Next, the equality $h_3(k^*;x^*)=0$ 
implies that 
\begin{align} \label{eq:h3}
    k_1^* 
        ~=~ \frac{(\kcat^* +\koff^* )\lcat^*  x^* _{10}}{  \kcat^*  x^* _1 x^* _2}
        ~=~ \frac{(\kcat^* +\koff^* )x_4^*  }{  x^* _1 x^* _2}~,
\end{align}
where the final equality follows from equation~\eqref{eq:k_9}. 
Thus, the expression for $k_1 $ given after~\eqref{eq:param-reduced-details} -- when evaluated at 
$(\kcat, \koff, \loff;x) =(\kcat^*, \koff^*, \loff^*;x^*) $
-- equals $k_1^*$.

Similarly,  the equality $h_4(k^*;x^*)=0$ 
(respectively, $h_5(k^*;x^*)=0$, $h_6(k^*;x^*)=0$, $h_7(k^*;x^*)=0$, or $h_{10}(k^*;x^*)=0$)
implies that 
$\kappa_3 $ (respectively, 
$m$, $\ell_1$, $\ell_3$, or $n$)
-- when evaluated at 
$(\kcat, \koff, \loff;x) =(\kcat^*, \koff^*, \loff^*;x^*)$
-- equals $k_3^*$ (respectively, $m^*$, $\ell^*_1$, $\ell^*_3$, or $n^*$).
Thus, $\phi (\kcat^*, \koff^*, \loff^*;x^*) = ( k^*;x^*)$.


The ``$\Leftarrow$'' direction is similar.  
Assume $\phi (\kcat^*, \koff^*, \loff^*;x^*) = (k^*;x^*)$.  
That is, the expressions for
 $k_1$, $k_3$, $m$, $\ell_1$, $\ell_3$, $\lcat$, and $n$
 evaluate to, respectively, $k_1^*$, $k_3^*$, $m^*$, $\ell_1^*$, $\ell_3^*$, $\lcat^*$, and $n^*$, 
 when $(\kcat, \koff, \loff;x) =(\kcat^*, \koff^*, \loff^*;x^*) $.
 In particular, equation~\eqref{eq:k_9} holds, and so $h_9(k^*;x^*)=0$.  Similarly, $h_i(k^*;x^*)=0$ for all other $i$ (here we also use equation~\eqref{eq:k_9}).
\end{proof}

\begin{remark} \label{rem:linearly-binomial}
The proof of Proposition~\ref{prop:param-reduced} proceeds by performing linear operations on the steady-state polynomials to yield binomials $g_i$, and then solving for one $k_j$ from each binomial to obtain the parametrization~\eqref{eq:param-reduced}.  This is similar in spirit to -- but more general than -- the approach prescribed in \cite[\S 4]{DPST} for ``linearly binomial'' networks.  Also, our linear operations were found ``by hand'', and so an interesting future direction is to develop efficient and systematic approaches to finding such operations leading to binomials.
\end{remark}

\begin{remark} \label{rem:toric-steady-states}
The proof of Proposition~\ref{prop:param-reduced} shows that
the ``steady-state ideal'' (the ideal generated by the right-hand sides of the ODEs)
of the reduced ERK network 
is generated by the binomials $g_i$.  
This network, therefore has, ``toric steady states''~\cite{TSS}.
In contrast, 
the steady-state ideal of the full ERK network is 
not a binomial ideal (it is straightforward to check this computationally, e.g., using the {\tt Binomials} package in {\tt Macaulay2}~\cite{M2}).
As for the irreversible versions of the ERK network, when the reactions with rate constants $\kon$ and $\lon$ are deleted, we see from~\eqref{eq:effectivefunction_irrev}
that the steady-state ideal becomes binomial. Hence, irreversible ERK networks that are missing both $\kon$ and $\lon$ are ``linearly binomial'' as in \cite[\S 4]{DPST}. 
\end{remark}



\begin{remark} \label{rmk:boundary-conservative}
All networks considered in this section are conservative, 
which can be seen from the conservation laws~\eqref{eq:conservation} for the full and irreversible ERK networks, and~\eqref{eq:cons-law-reduced} for the reduced ERK network.
Also for these networks, there are no boundary steady states in any compatibility class 
(it is straightforward to check this using results from~\cite{PetriNetExtended} or~\cite{ShiuSturmfels}).
\end{remark}

\section{Main Results} \label{sec:results}
Each ERK network we investigated admits oscillations via a Hopf bifurcation (Section~\ref{sec:oscillations}). 
Bistability, however, is more subtle (Section~\ref{sec:bistability}).

\subsection{Oscillations} \label{sec:oscillations}
 The full ERK system (Figure \ref{fig:erk_sys}) exhibits oscillations for some values of the rate constants~\cite{long-term}. We now investigate oscillations in the fully irreversible and reduced ERK networks. 

\subsubsection{Fully irreversible ERK network} \label{sec:osc-irrev}
As shown in Figure~\ref{fig:oscill-in-irr-erk}, 
the fully irreversible ERK network admits oscillations. 
That figure was generated using the following rate constants:  
\begin{align} \label{eq:rates-osc-irr} 
 (k_1  , k_3 , \kcat  , \koff , \ell_1  , \ell_3 , \lcat  , \loff ,  m_2  , m_3  , n_1  , n_3 )
 ~=~  
 &  
( 5241, 5314.5, 1291, 76.203,   64.271, 
\\  \notag
    & \quad 
    44.965, 924970, 27238, 2.76250\times 10^6,  
\\ \notag 
 & \quad
 2.0451,  2.1496\times 10^6,  1.3334 )~.
\end{align}
These rate constants~\eqref{eq:rates-osc-irr} come from the ones that 
Rubinstein {\em et al.} showed generate oscillations for the full ERK network \cite[Table 2]{long-term} (we simply ignore their rate constants for the six deleted reactions).  The approximate initial species concentrations used to generate Figure~\ref{fig:oscill-in-irr-erk} are as follows (see supplementary file {\tt ERK-Matcont.txt}):
\begin{align} \label{eq:fig-inital-cond} \notag
(x_1,x_2, \dots, x_{12}) 
 ~\approx~  
  ( &1.215\times 10^{-5},  ~ 4.722\times 10^{-5},  ~ 8.777\times 10^{-4},  ~ 1.396\times 10^{-3}, 
    \\ 
    & 
    6.590\times 10^{-8}, ~ 2.698\times 10^{-3},  ~ 2.873\times 10^{-4},  ~ 1.150\times 10^{-3},  
    \\
    & 
    3.072\times 10^{-3},  ~ 2.262\times 10^{-6}, ~ 0.042,  ~ 0.849)~. \notag
\end{align}

\begin{figure}[ht]
    \begin{minipage}[t]{.3\textwidth}
      \includegraphics[width=\linewidth]{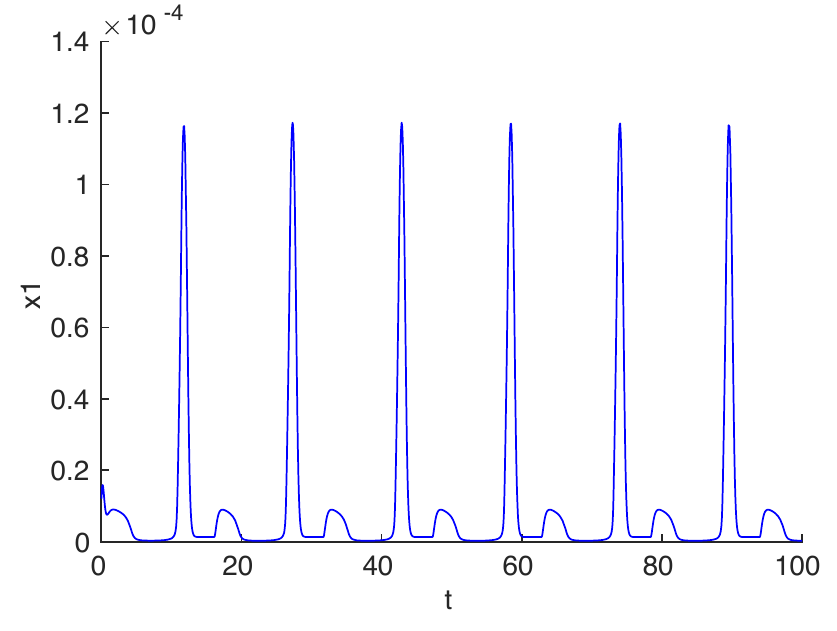}
    \end{minipage}
    \hfill
    \begin{minipage}[t]{.3\textwidth}
      \includegraphics[width=\linewidth]{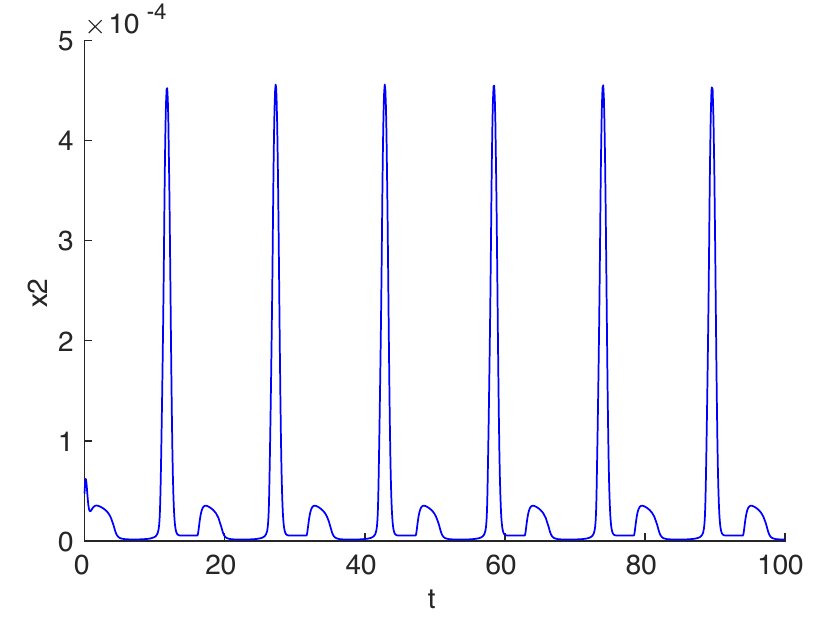}
    \end{minipage}
    \hfill   
    \begin{minipage}[t]{.3\textwidth}
      \includegraphics[width=\linewidth]{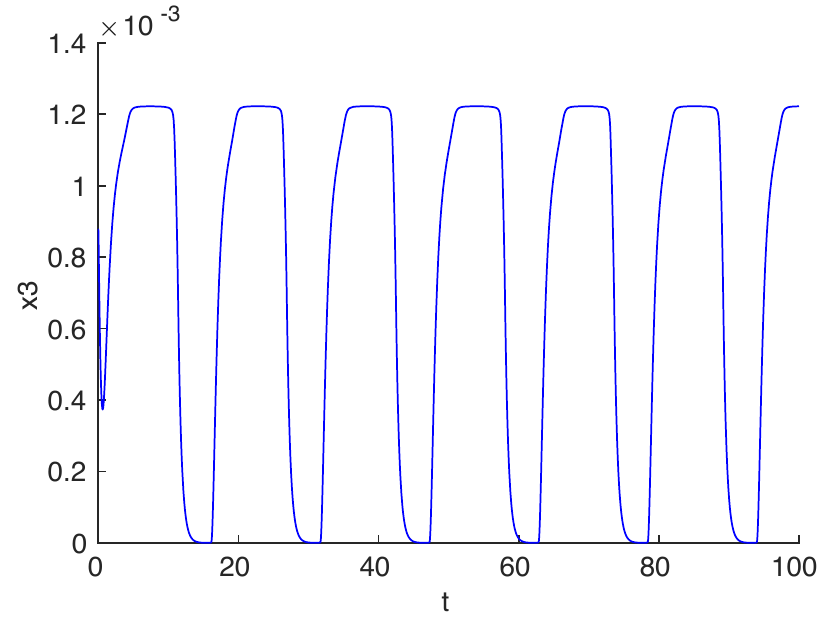}
    \end{minipage}
    \hfill
    \begin{minipage}[t]{.3\textwidth}
      \includegraphics[width=\linewidth]{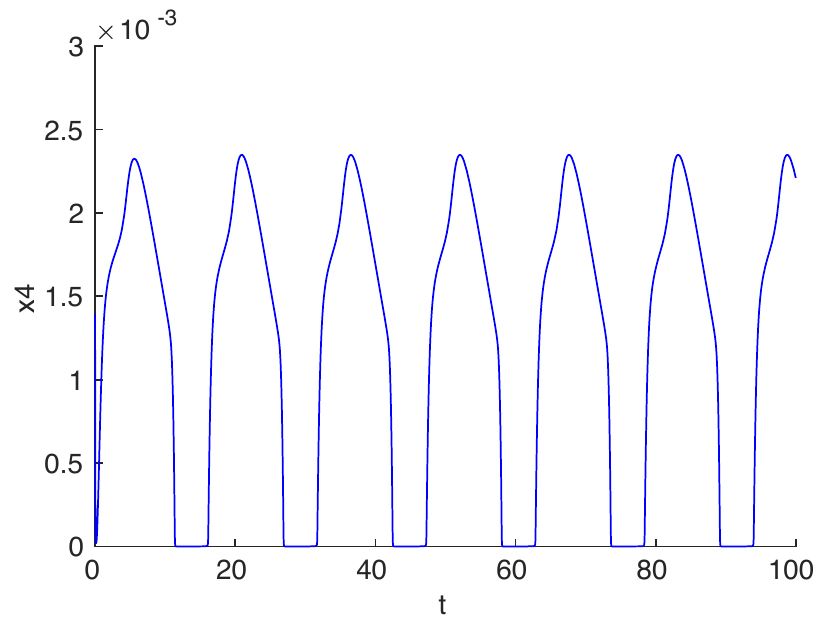}
    \end{minipage}
    \hfill    
    \begin{minipage}[t]{.3\textwidth}
      \includegraphics[width=\linewidth]{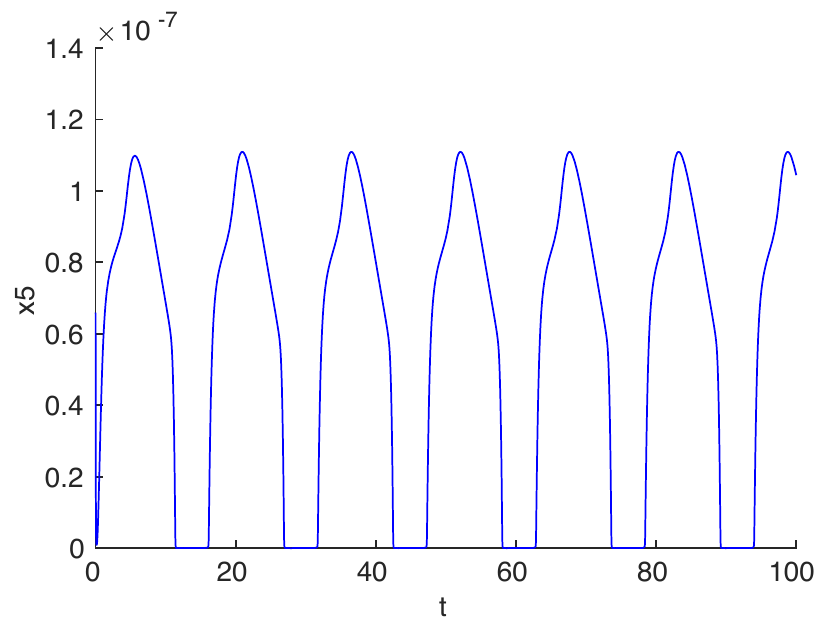}
    \end{minipage}
    \hfill
    \begin{minipage}[t]{.3\textwidth}
      \includegraphics[width=\linewidth]{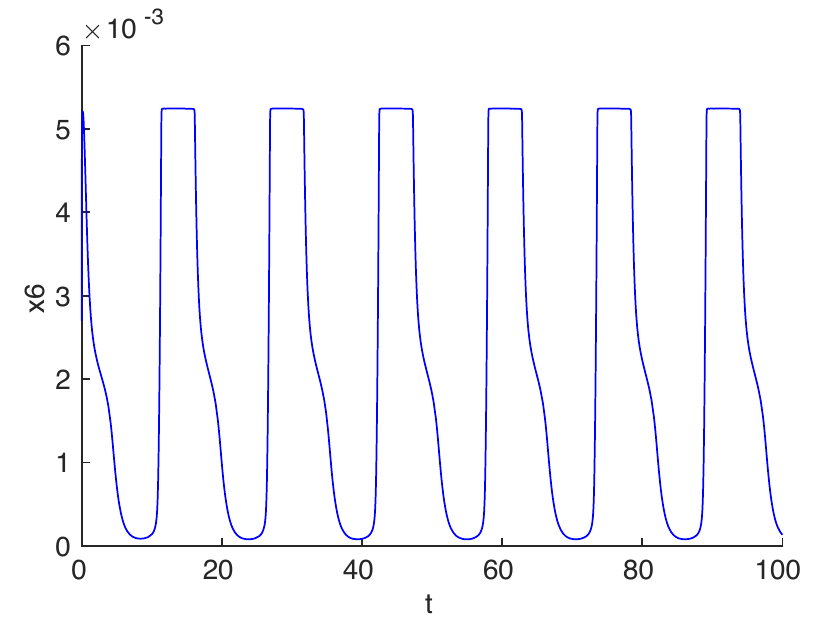}
    \end{minipage}
    \hfill    
    \begin{minipage}[t]{.3\textwidth}
      \includegraphics[width=\linewidth]{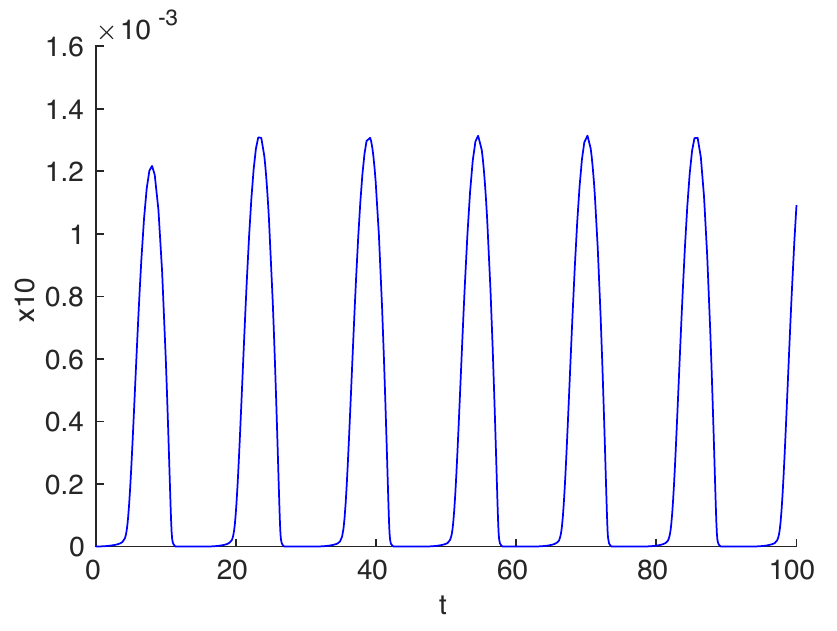}
    \end{minipage}
    \hfill
    \begin{minipage}[t]{.3\textwidth}
      \includegraphics[width=\linewidth]{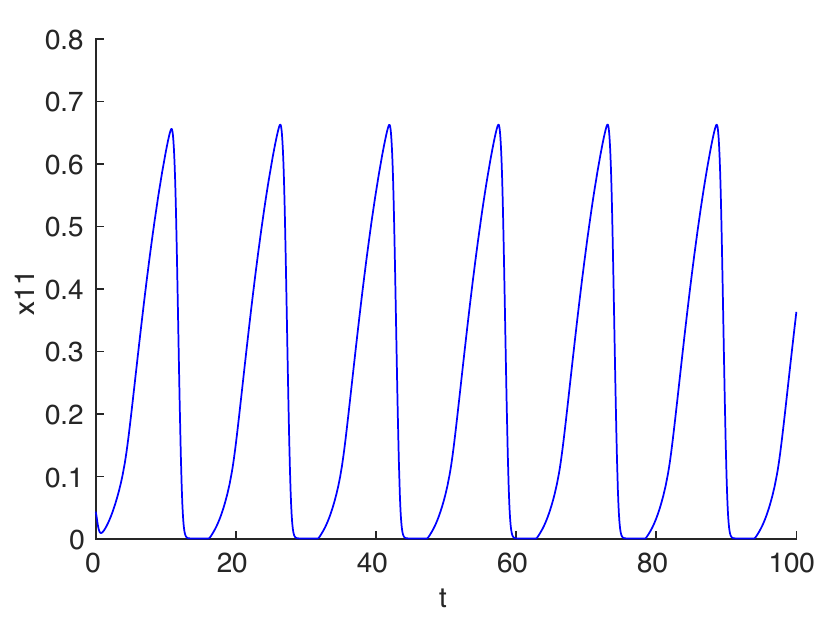}
    \end{minipage}
    \hfill
    \begin{minipage}[t]{.3\textwidth}
      \includegraphics[width=\linewidth]{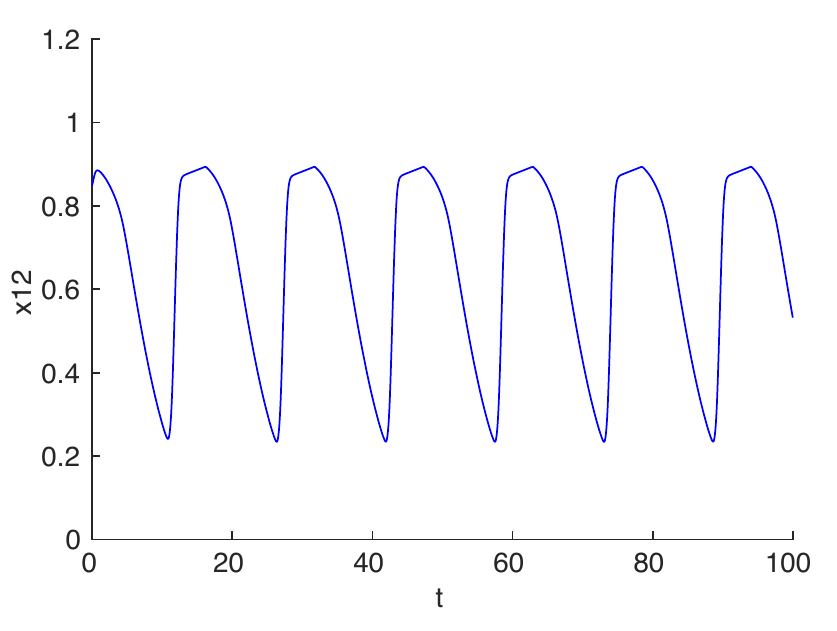}
    \end{minipage}
    \caption{The fully irreversible ERK network undergoes oscillations 
    when the rate constants are as in~\eqref{eq:rates-osc-irr}
    and the initial species concentrations are as in~\eqref{eq:fig-inital-cond}.
    Displayed in this figure are all species concentrations, except $x_7, x_8,$ and $x_9$. 
    This figure was generated using {\tt MATCONT}, a numerical bifurcation package~\cite{matcont}.  For details, see the supplementary file~{\tt ERK-Matcont.txt}.}
    \label{fig:oscill-in-irr-erk}
\end{figure}

In Figure~\ref{fig:oscill-in-irr-erk}, we notice some peculiarities in the graphs $x_i(t)$ of the species concentrations. 
The species concentrations $x_1$ and $x_2$ (corresponding to \ce{S00} and \ce{E}, respectively) peak dramatically, while $x_3$ and $x_6$ (\ce{F} and~\ce{S01F}) stabilize momentarily at each peak.
Also, each of $x_1,x_2,x_3,x_4,x_5,x_{10},x_{11}$ deplete for some time in each period, whereas $x_{12}$ (\ce{S11}) never depletes. Finally, the graphs of the pairs $x_1$ and $x_2$ are qualitatively similar, and also the pair $x_3$ and $x_6$, the pair $x_4$ and $x_5$, and the pair $x_{10}$ and $x_{11}$.

Going beyond the fully irreversible ERK network, all other irreversible ERK networks 
 -- those obtained from the full ERK network 
by deleting one or more the reactions $k_2, \kon, m_1, \ell_2, \lon, n_2$ --  
also admit oscillations.  This claim follows from a result of Banaji that ``lifts'' oscillations when one or more reactions are made reversible \cite[Proposition 4.1]{banaji-inheritance}.

\subsubsection{Reduced ERK network} \label{sec:osc-reduced}
We saw in the previous subsection that the fully irreversible ERK network exhibits oscillations. We now show that a simpler network - the reduced ERK network - also undergoes oscillations via a Hopf bifurcation. 
These oscillations are shown in
Figure~\ref{fig:oscill-in-red-erk}, and the rate constants 
that yield the corresponding Hopf bifurcation are specified in Theorem~\ref{thm:hopf}. 

Compared to the oscillations for the irreversible ERK network (Figure~\ref{fig:oscill-in-irr-erk}), 
the oscillations in the reduced ERK network (Figure~\ref{fig:oscill-in-red-erk}) are more uniform.
Also, the period of oscillation is much shorter, and the amplitudes for species $x_3$, $x_8$, and $x_{10}$ are small (this may be due to the choice of rate constants). 
Finally, three of the six species shown do not deplete completely, whereas nearly all the species of the fully irreversible ERK do deplete in each period. 


We discovered oscillations by finding a Hopf bifurcation.  
How we found this bifurcation -- via the Hopf-bifurcation criterion in Section~\ref{sec:param-Hopf} --
is the focus of the rest of this subsection.  


\begin{figure}[ht]
    \begin{minipage}[t]{.3\textwidth}
      \includegraphics[width=\textwidth]{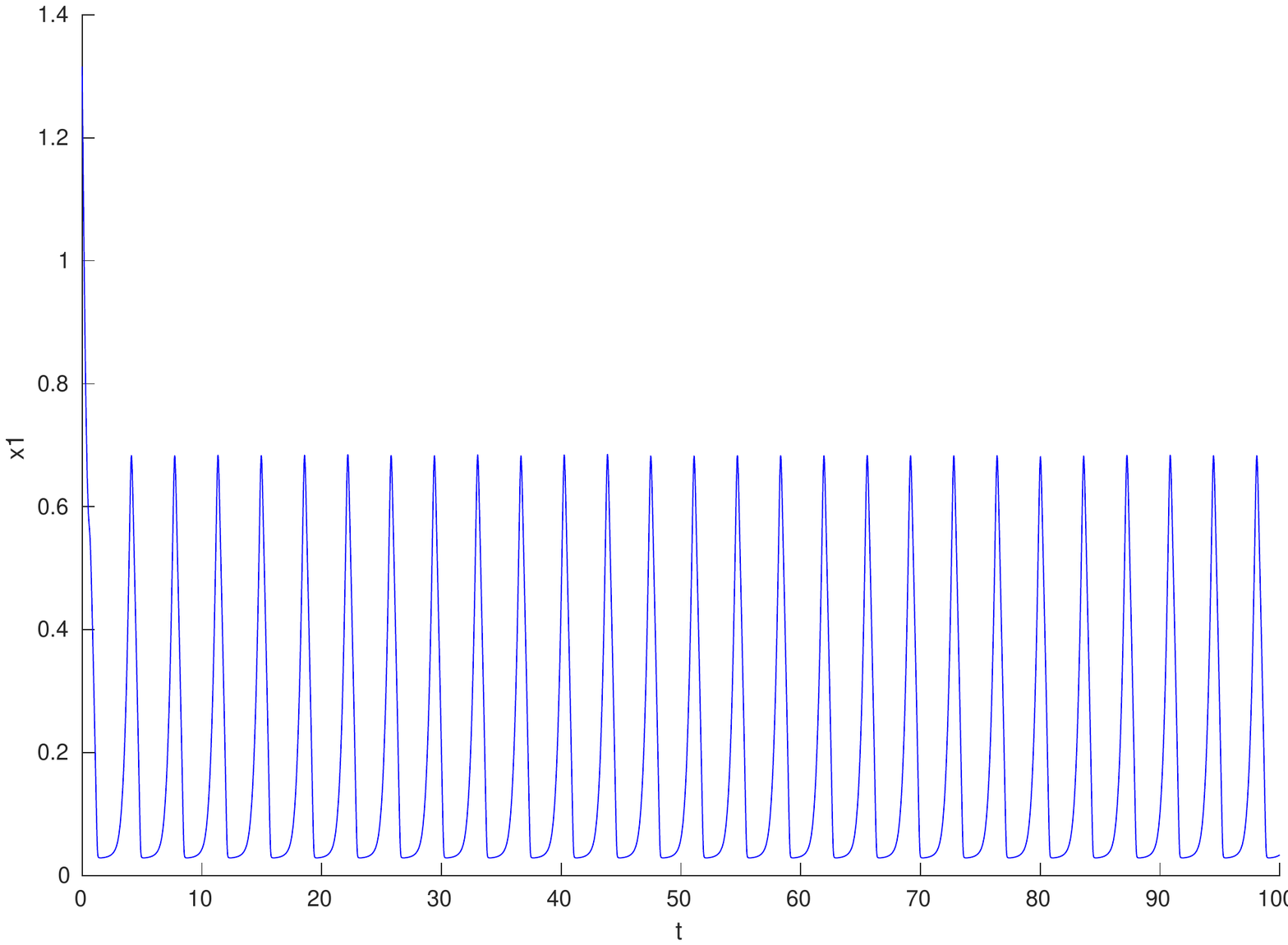}
    \end{minipage}%
    \begin{minipage}[t]{.3\textwidth}
      \includegraphics[width=\textwidth]{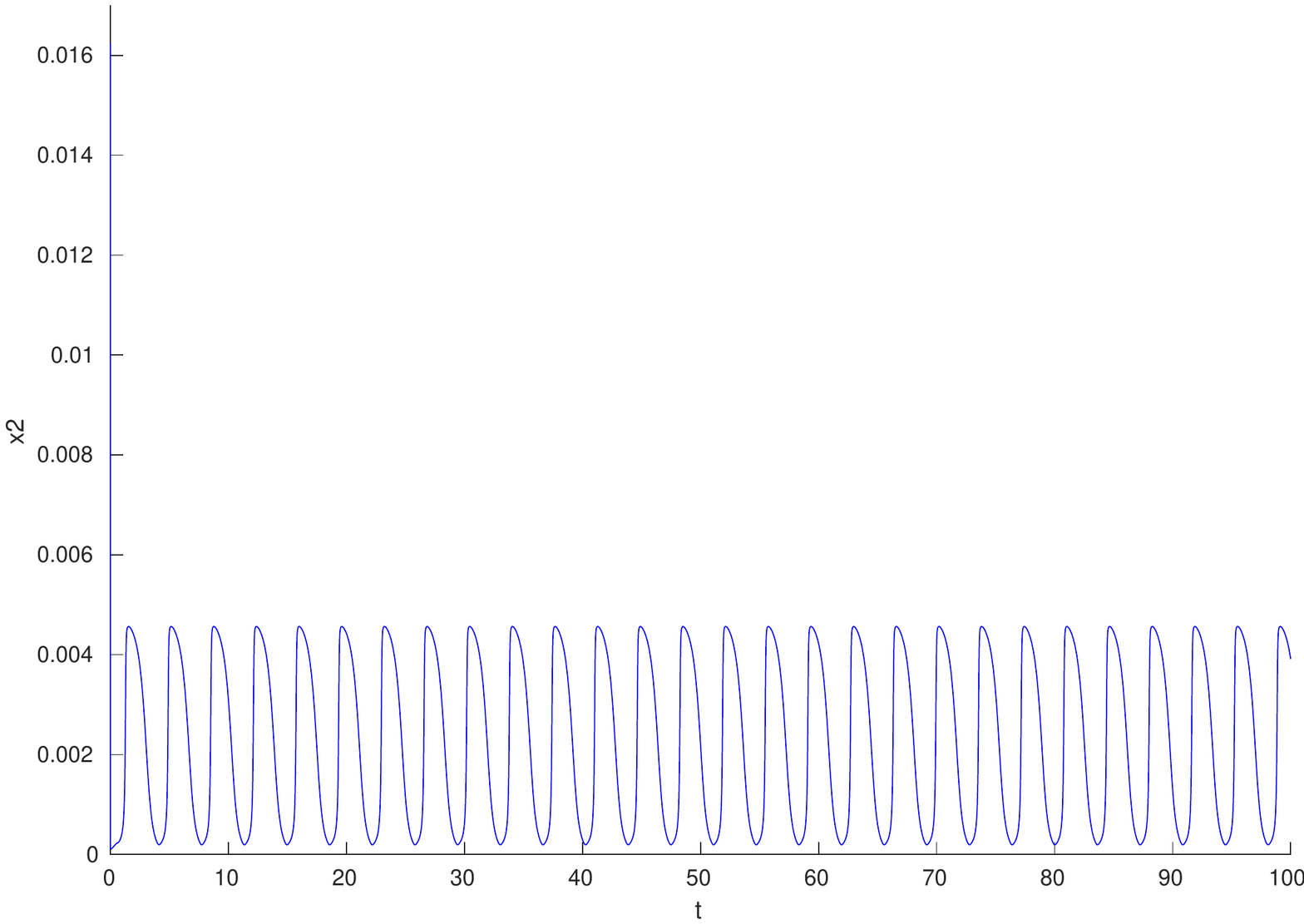}
    \end{minipage}%
    \begin{minipage}[t]{.3\textwidth}
      \includegraphics[width=\textwidth]{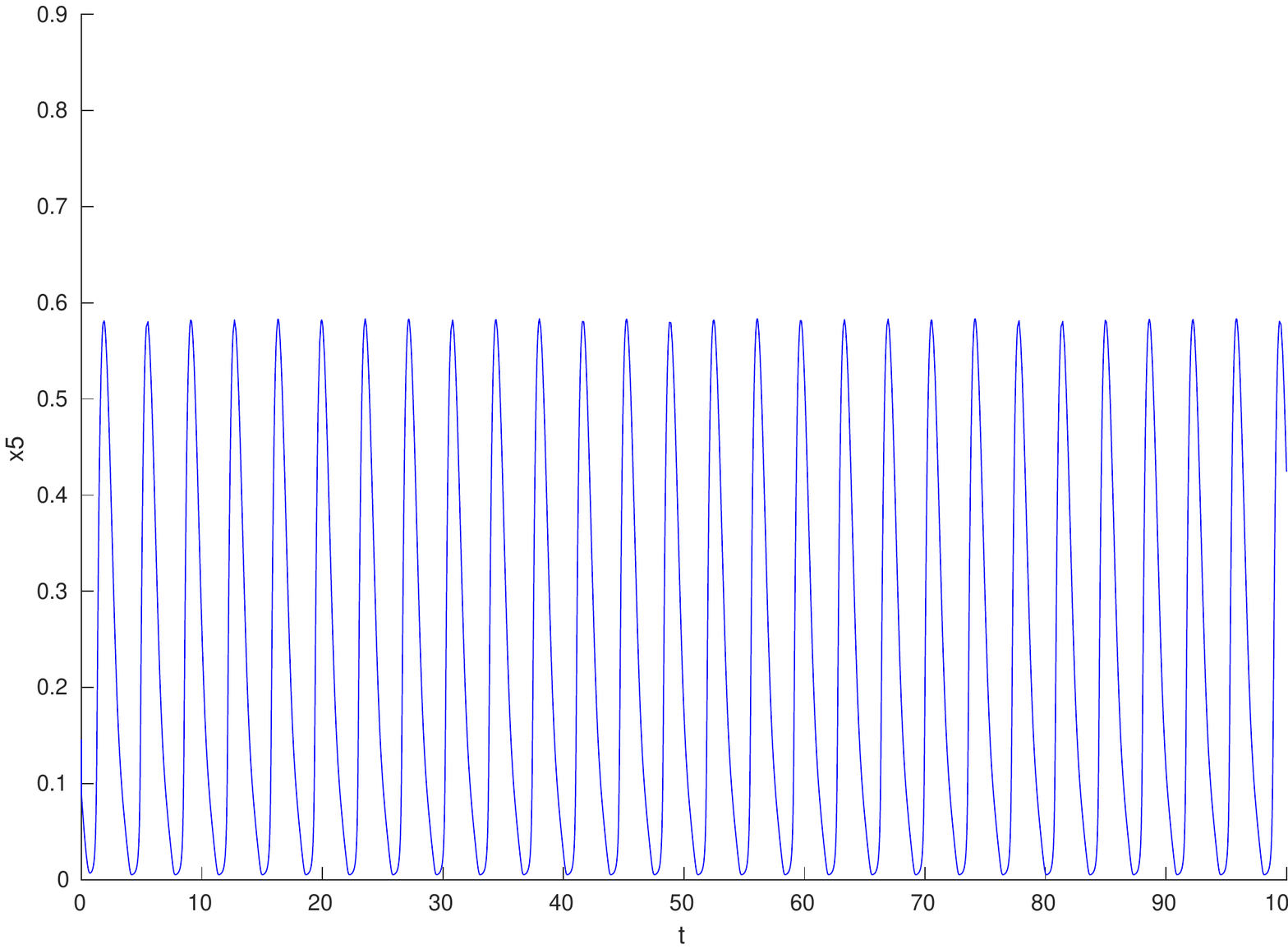}
    \end{minipage}%
    \vspace{-1in}
    \begin{minipage}[t]{.3\textwidth}
      \includegraphics[width=\textwidth]{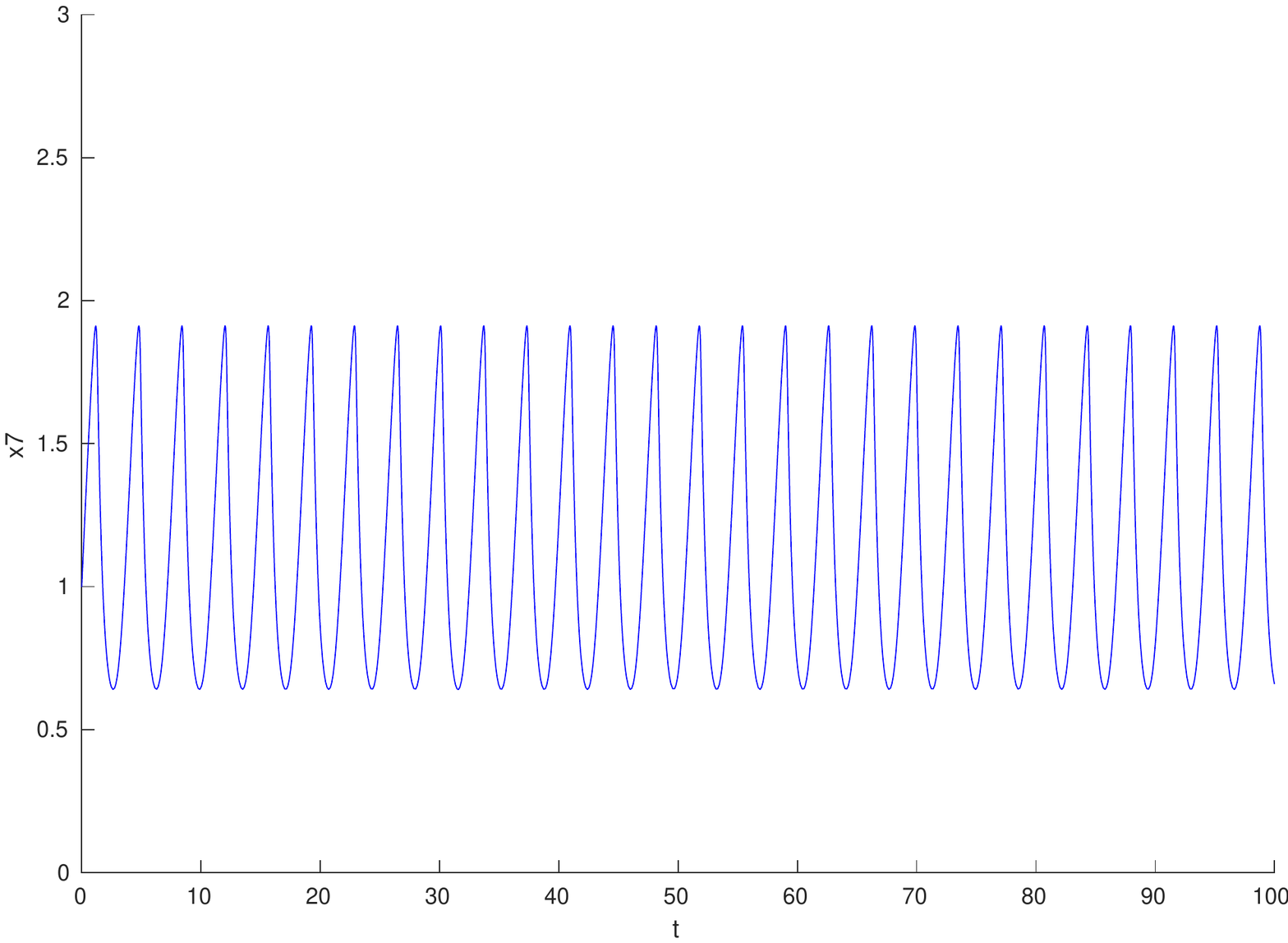}
    \end{minipage}%
    \begin{minipage}[t]{.3\textwidth}
      \includegraphics[width=\textwidth]{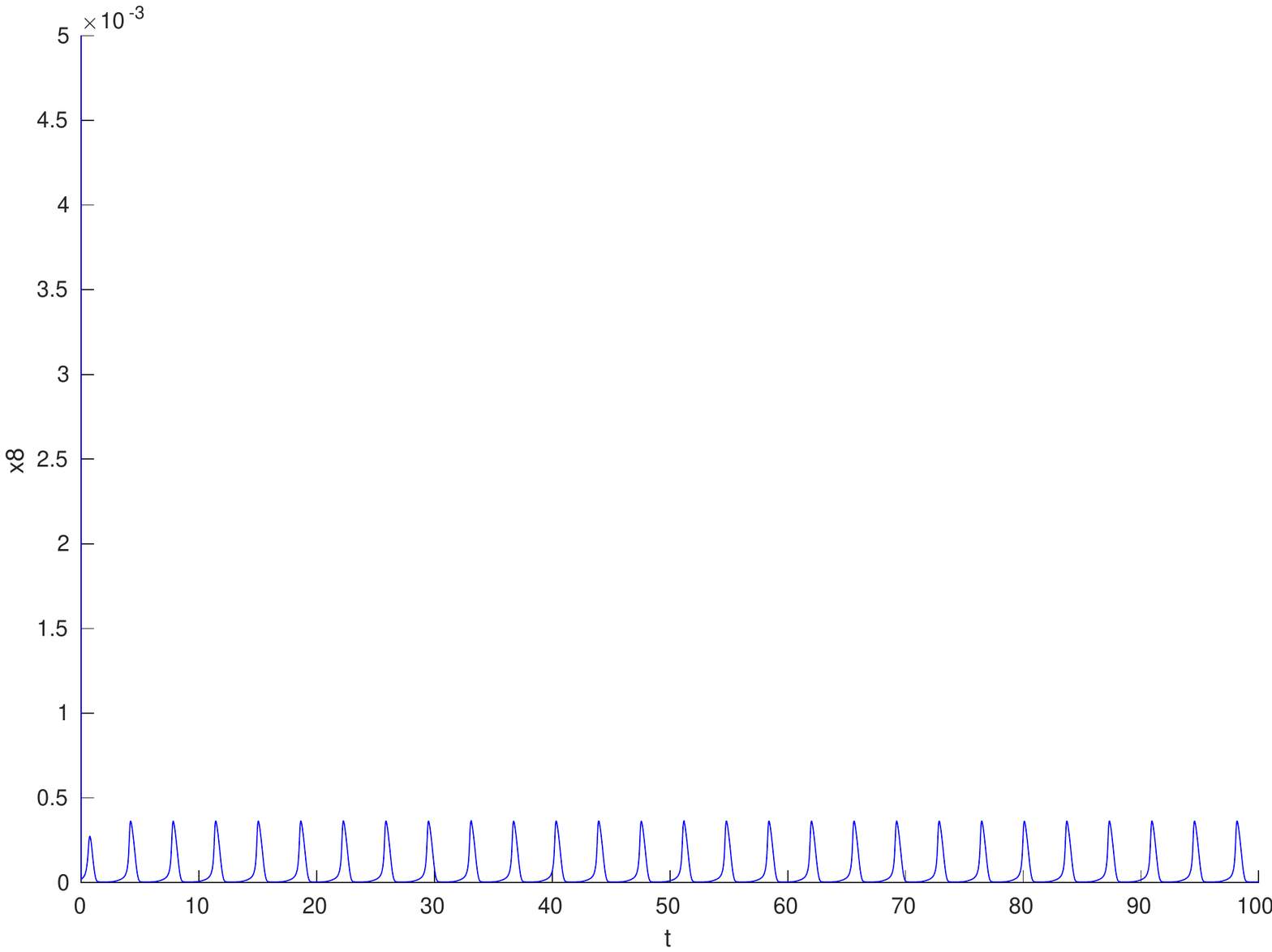}
    \end{minipage}%
    \begin{minipage}[t]{.3\textwidth}
      \includegraphics[width=\textwidth]{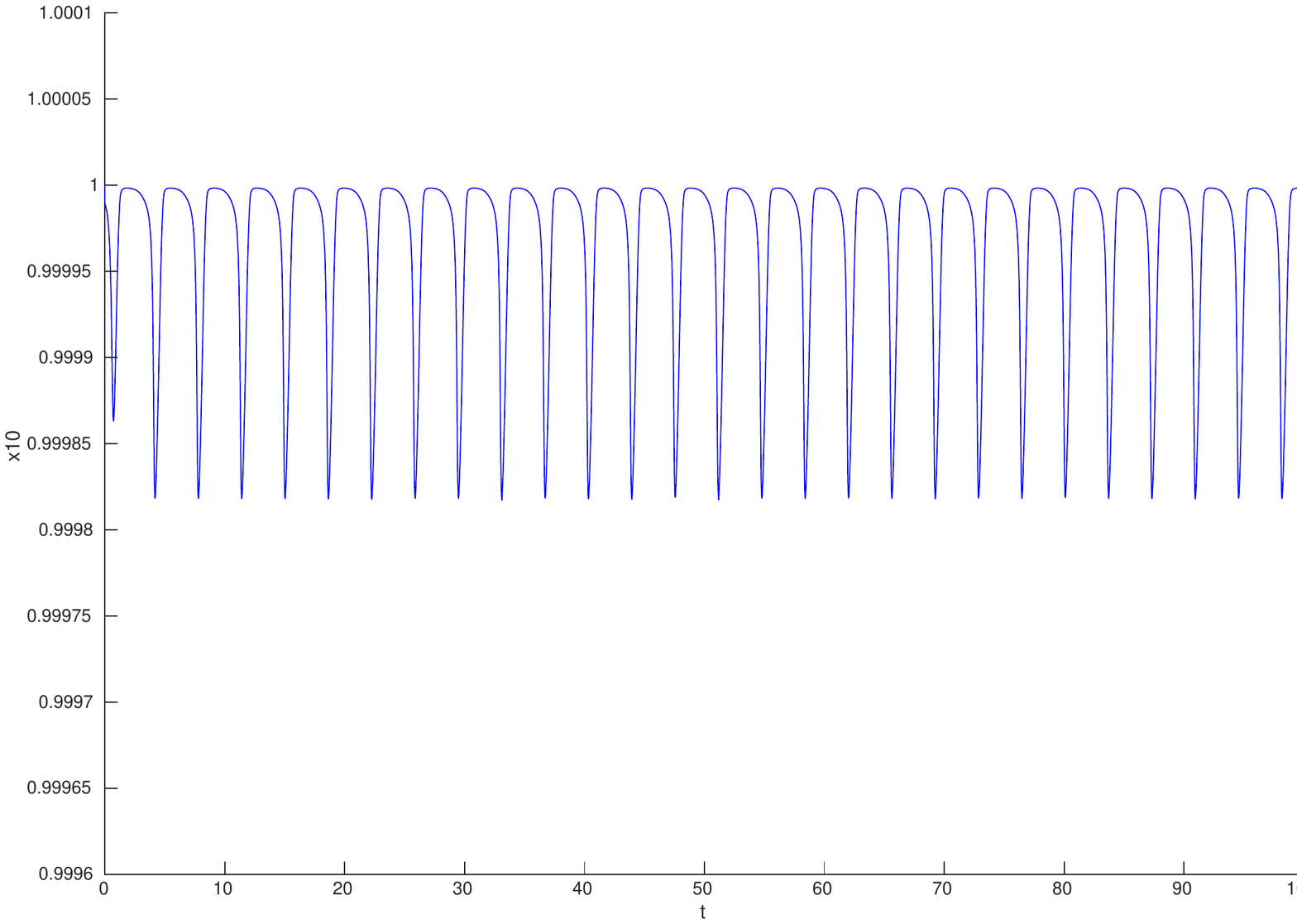}
    \end{minipage}
    \caption{The reduced ERK network exhibits oscillations 
    when the rate constants are 
    approximately those  
    in Theorem~\ref{thm:hopf} and the initial species concentrations are close to the Hopf bifurcation.
    Details are in the supplementary file~{\tt ERK-Matcont.txt}.
    This figure, generated using {\tt MATCONT}, displays all species concentrations, except $x_3, x_4, x_6,$ and $x_9$.}
    \label{fig:oscill-in-red-erk}
\end{figure}

\begin{proposition}[Hopf criterion for reduced ERK] \label{prop:hopfcriterion}
Consider the reduced ERK network, 
and let the  
polynomials $f_i$ denote the
right-hand sides of the resulting ODEs, as in~\eqref{eq:ODE-reduced}.
Let $\hat{\kappa}:=(\kcat, \koff, \loff)$ and $x:=(x_1, x_2, \dots, x_{10})$, and 
let $\phi$ be the steady-state parametrization~\eqref{eq:param-reduced}. 
Then the following is a univariate, degree-7 polynomial in $\lambda$, with coefficients in $\mathbb{Q}(x)[\hat{\kappa}]$:
 \begin{align} \label{eq:reduced-char-poly}
    q(\lambda) ~:=~
        \frac{1}{\lambda^3} ~ \det \left( \lambda I - {\rm Jac} (f) \right) |_{(\kappa; x)= \phi(\hat{\kappa}; x)} ~.
 \end{align}
Now let $\mathfrak{h}_i$, for $i=4,5,6$, denote the determinant of the $i$-th Hurwitz matrix of the polynomial~$q(\lambda)$ in~\eqref{eq:reduced-char-poly}. 
Then the following are equivalent:
\begin{enumerate}[(1)]
    \item there exists a rate-constant vector $\kappa^* \in \mathbb{R}^{10}_{>0}$ such that the resulting system~\eqref{eq:ODE-reduced} exhibits a simple Hopf bifurcation, with respect to $\kcat$, at some 
    $x^* \in \mathbb{R}^{10}_{>0}$, and
    \item there exist  $x^* \in \mathbb{R}^{10}_{>0}$ and $\hat{\kappa}^*\in \mathbb{R}^3_{>0}$ such that 
    \begin{align} \label{eq:hurwitz-hopf-conditions-reduced-ERK}
   \mathfrak{h}_4 (\hat{\kappa}^*; x^*) >&0~, ~ 
    \mathfrak{h}_5 (\hat{\kappa}^*; x^*) >0~, ~ 
    \mathfrak{h}_6 (\hat{\kappa}^*; x^*) =0~, ~{\rm and} ~\\ \notag
    \frac{\partial}{\partial \kcat }& \mathfrak{h}_6 (\hat{\kappa}; x) |_{(\hat{\kappa}; x)=(\hat{\kappa}^*; x^*)} \neq 0~.
    \end{align}
\end{enumerate}
Moreover, given $\hat{\kappa}^*$ and $x^*$ as in (2), a simple Hopf bifurcation with respect to $\kcat$ occurs at $x^*$ when the rate constants are taken to be $\kappa^*:= \widetilde{\pi} (\phi(\hat{\kappa}^*; x^*))$.  
Here, $\widetilde{\pi}:  \mathbb{R}_{>0}^{10}\times  \mathbb{R}_{>0}^{10} \to \mathbb{R}_{>0}^{10}$ is the natural projection to the first 10 coordinates.
 \end{proposition}

\begin{proof} 
The fact that $q(\lambda)$ is a degree-7 polynomial follows from Lemma~\ref{lem:reduced-char-poly-general}, and the fact that its coefficients are in $\mathbb{Q}(x)[\hat{\kappa}]$ follows from inspecting equations~\eqref{eq:ODE-reduced} and~\eqref{eq:param-reduced}.  The rest of the result will follow immediately from Theorem~\ref{thm:hopf-criterion} and Proposition~\ref{prop:param-reduced}, once we prove that $\mathfrak{h}_1$,
$\mathfrak{h}_2$,
$\mathfrak{h}_3$, and the constant term of $q(\lambda)$ are all positive when evaluated at any $(\hat{\kappa}; x) \in \mathbb{R}^3_{>0} \times \mathbb{R}^{10}_{>0}$.  Indeed, this is shown in the 
supplementary file
{\tt reducedERK-hopf.mw}.  (In fact, even before substituting the parametrization 
$(\kappa; x)=\phi(\hat{\kappa}; x)$, the corresponding Hurwitz determinants are already positive polynomials.)
\end{proof}

\begin{remark}
Note that $\kcat$ is the only free parameter, so it is the natural bifurcation parameter.
\end{remark}

We now prove that the reduced ERK network gives rise to a Hopf bifurcation.

\begin{theorem}[Hopf bifurcation in reduced ERK] \label{thm:hopf}
The reduced ERK network exhibits a simple Hopf bifurcation with respect to the bifurcation parameter $\kcat$ at the following point:
\[
x^*~\approx~ (
0.05952457867,~ 
0.002204614024,~ 
1,~ 
1,~ 
0.1518056972,~ 
1,~ 
1,~ 
0.00001239529511,~ 
1,~ 
1 
)~,
\]
when the rate constants are as follows:
\begin{align*}
\begin{split}
    (k^*_1, k^*_3, \kcat^*, \koff^*, m^*, \ell^*_1, \ell^*_3, \lcat^*, \loff^*, n^*)~\approx~
   & (5.562806640\times10^6,~ 
730, ~
729, ~
1, ~
453.5941390, ~
\\
&3.879519315\times10^8,~
730,~
729,~
1,~
80675.77183 
)
\end{split}
\end{align*}

Here, 
$\phi$ is the parametrization~\eqref{eq:param-reduced}, and 
$\widetilde{\pi}$ 
is the projection to the first 10 coordinates.
\end{theorem}
\begin{proof}
By Proposition~\ref{prop:hopfcriterion}, we need only show that 
the inequalities and equality in~\eqref{eq:hurwitz-hopf-conditions-reduced-ERK} are satisfied at
$x=x^*$ (with $x^*$ given in the statement of the theorem) and $\hat{\kappa}=\hat{\kappa}^*=(9,1,1)$.  
These are verified in the supplementary file 
{\tt reducedERK-hopf.mw}.
\end{proof}
 
\begin{remark} \label{rmk:point-to-appendix}
The Hopf bifurcation given in Theorem~\ref{thm:hopf} was found by analyzing the Newton polytopes of $\mathfrak{h_4}$, $\mathfrak{h_5}$, and $\mathfrak{h_6}$.  The theory behind this approach is presented in Appendix~\ref{app:cones}, and the steps we took to find the Hopf bifurcation are listed in Appendix~\ref{sec:appendix-using-cones-method}.  
We include these appendices for readers who wish to apply similar approaches to other systems.
\end{remark}

\subsection{Bistability} \label{sec:bistability}
Although the full ERK network is bistable~\cite{long-term}, we now prove that 
the reduced ERK network is not bistable (Proposition~\ref{prop:bistable}).
As for irreversible ERK networks, 
some of them are bistable, and 
we show that bistability is controlled by the two reactions $\kon$ and $\lon$ (Theorem~\ref{thm:bistable}).

\begin{proposition}\label{prop:bistable} 
The reduced ERK network is not multistationary, and hence not bistable. 
\end{proposition}

\begin{proof}
Let $\mathcal{N}$ denote the reduced ERK network.  
By definition and Proposition~\ref{prop:param-reduced}, we obtain the following critical function for $\mathcal N$:
\begin{equation} \label{eq:critical-function-proof}
 C(\hat a; x) \quad  = \quad \left(\det {\rm Jac}~\augmentH\right)|_{(a; x)=\phi(\hat a;  x)}~,
\end{equation}
where $\hat a = (\kcat, \koff, \loff)$, 
the function
$\augmentH$ is as in~\eqref{eq:effectivefunction_reduced}, 
and $\phi(\hat a; x)$ is as in~\eqref{eq:param-reduced}.

This critical function $ C(\hat a; x)$ (see the supplementary file 
{\tt reducedERK-noMSS.mw})
is a rational function,
where the denominator is the following monomial:
$x_1 x_2 x_3 x_5 x_6 x_7 x_8 x_9$. 
The numerator of $C(\hat a;  x)$  is the following polynomial, which is negative when evaluated at any $(\hat a; x) \in \mathbb{R}_{>0}^3\times \mathbb{R}_{>0}^{10}$:
\begin{align*}
    & -\kcat^3 (\kcat+\koff)^2 x_4^3 (\kcat x_4+\loff x_{10})^2 \loff \koff 
    (x_1 x_2 x_8 + x_1 x_3 x_8 + 
    x_1 x_4 x_8 +  x_{10} x_2 x_5 \\
    & \quad \quad 
    + x_{10} x_2 x_6 + x_{10} x_2 x_8 + x_2 x_3 x_8 + x_2 x_4 x_8 + x_2 x_5 x_8 + x_2 x_5 x_9 + x_2 x_6 x_8 + x_2 x_6 x_9 
     \\
    & \quad \quad 
    + x_2 x_7 x_8 
    + x_2 x_8 x_9 + x_3 x_7 x_8 + x_4 x_7 x_8 )~.
\end{align*}
Thus, the following holds 
for all $(\hat a; x) \in \mathbb{R}_{>0}^3\times \mathbb{R}_{>0}^{10}$:
\[ {\rm sign}(C(\hat a;  x))= -1 = (-1)^{{\rm rank}(N)}~,\]
where the final equality uses the fact that the stoichiometric matrix $N$ has rank $10-3=7$.

So, by Proposition~\ref{prop:c-general} and the fact that
$\mathcal N$ is conservative with no boundary steady states in any stoichiometric compatibility class (Remark~\ref{rmk:boundary-conservative}), $\mathcal N$ is monostationary.  Thus, $\mathcal N$ is {\em not} multistationary and so, by definition, is {\em not} bistable.
\end{proof}

Although the reduced ERK network is not bistable (Proposition~\ref{prop:bistable}),
the next result shows that irreversible versions of the full ERK network
are bistable, as long as one of the reactions labeled by $\kon$ and $\lon$ is present.  
That is, this result tells us which reactions can be safely deleted (in contrast to standard results 
concerning reactions that can be added) while preserving bistability.

\begin{theorem}[Bistability in irreversible ERK networks] \label{thm:bistable}
Consider any network $\mathcal N$ obtained from the full ERK network  
by deleting one or more of the reactions 
corresponding to rate constants $k_2, \kon, m_1, \ell_2, \lon, n_2$ (blue in Figure~\ref{fig:erk_sys}).
Then the following are equivalent:
\begin{enumerate}[(1)]
    \item $\mathcal N$ is multistationary,  
    \item $\mathcal N$ is bistable, and 
    \item $\mathcal N$ contains at least one of the reactions labeled by $\kon$ and $\lon$.
\end{enumerate}
  \end{theorem} 
\begin{proof}
By definition, every bistable network is multistationary, so $(2)\Rightarrow (1)$. We therefore need only show $(1)\Rightarrow (3) \Rightarrow (2)$.  (All computations below are found in our supplementary file {\tt irreversibleERK.mw}).

For $(1)\Rightarrow (3)$, we will prove $\neg (3)\Rightarrow \neg (1)$: 
Assume that $\mathcal{N}$ contains neither the reaction labeled by $\kon$ nor the reaction $\lon$.  Our proof here is analogous to that of Proposition~\ref{prop:bistable}. 
By Proposition~\ref{prop:param-irrev}, we obtain a critical function, $ C(\hat a;  x)$, for $\mathcal N$ 
of the form~\eqref{eq:critical-function-proof}, where now  $\augmentH$ is as in~\eqref{eq:effectivefunction_irrev} (with $\mathds{1}_{\kon} = \mathds{1}_{\lon}=0$)
and $\phi(\hat a; x)$ is as in~\eqref{eq:param-irrev-details} (with $\hat a = (a_2,a_4)$).  

Here, $\det {\rm Jac}(\augmentH)$ is a rational function with denominator equal to $\koff x_2(n_2+n_3)\lcat l_3 k_3 m_3$, which is always positive.
The numerator is a polynomial of degree 5 in the variables $x_2,x_3$, and $x_9$ with coefficients that are always negative (see the supplementary file). 
The critical function 
 $ C(\hat a;  x)$
 is obtained by substituting the positive parametrization into 
$\det {\rm Jac}(\augmentH)$. 
Hence, 
for all $(\hat a; x) \in \mathbb{R}_{>0}^2\times \mathbb{R}_{>0}^{12}$, the equality ${\rm sign} (C(\hat a;  x)) = -1 = (-1)^{{\rm rank}(N)}$ holds, because the stoichiometric matrix $N$ has rank $12-3=9$.
So, by Proposition~\ref{prop:c-general} (recall from Remark~\ref{rmk:boundary-conservative} that $\mathcal N$ is conservative with no boundary steady states in any stoichiometric compatibility class), $\mathcal N$ is 
{\em not} multistationary.

Now we show $(3) \Rightarrow (2)$, that is, if $\mathcal{N}$ contains at least one of the reactions labeled by $\kon$ and $\lon$ then $\mathcal{N}$ is bistable.  
By symmetry (from exchanging in the network $E$, $S_{00}$, and $S_{01}$ with, respectively, $F$, $S_{11}$, and $S_{10}$), we may assume that $\mathcal N$ contains $\kon$. 

Consider the network $\mathcal{N}'$ obtained from the full ERK network
by deleting all reactions marked in blue in Figure~\ref{fig:erk_sys}, except for $\kon$ (equivalently, we set $k_2=m_1=\ell_2=\lon=n_2=0$).  
We will show that the following total constants and rate constants yield bistability: 
\begin{align} \label{eq:specific-params}
 & \quad \quad \quad \quad (c_1 ,~ c_2, ~ c_3) ~=~(46, 13,  13)~, \quad {\rm and}  \notag \\ 
   &(k_1,
    	k_3, 
	\kcat, 
	\kon, 
	\koff, 
    	{\ell}_1,
	{\ell}_3, 
	\lcat, 
	\loff, 
    	m_2, 
	m_3,
    	n_1,  
	n_3
    ) 
    ~=~ \\ \notag
    & \quad \quad  \left(    
  	2, 
    	1.1, 
	1,  
	5, 
	15, 
    	2, 
	1.1, 
	1, 
	10, 
    	20, 
	10, 
    	20, 
	10 
    \right )~. 
\end{align}


Among the resulting three steady states (see the supplementary file), one of them is approximately: 
\begin{align*}
\begin{tabular}{llllll}
 (20.72107755, & 0.2956877203, &3.248789181,& 7.821850626, &0.7821850626, &1.147175131, \\ 
 0.7821850626, & 0.7821850626, &0.1765542587, &1.322653950,& 11.13994215, &1.324191138)~.
\end{tabular}
\end{align*}

At the above steady state, the Jacobian matrix (of the system obtained from~\eqref{eq:ODE-full} by making the substitutions~\eqref{eq:specific-params} and $ k_2=m_1=\ell_2=\lon=n_2=0$) has three zero eigenvalues (due to the three conservation laws).  For the remaining eigenvalues, the real parts are approximately:
\begin{align*}
\begin{tabular}{lll}
-76.0913958200572,~&-70.7106617930401 ,~&-16.3022723748274,\\ 
-10.9324829878475,~ &-10.9324829878475,~&-8.81318904794782 ,\\ 
-4.88866989801728,~&-4.88866989801728 ,~&-0.0545784672515179 ,
\end{tabular}
\end{align*}
Thus, the nonzero eigenvalues have strictly negative real part, so the steady state is exponentially stable.



Another steady state is approximately
\begin{align*}
\begin{tabular}{llllll}
(0.1782157709, & 8.088440520, & 0.2275355904, &11.45336411,& 1.145336411,& 0.1737638914,\\ 
 1.145336411,& 1.145336411, &0.3818389270, &0.07080081803, &2.620886659, &27.68512059)~.
    \end{tabular}
\end{align*}
At this steady state, the real part of the eigenvalues of the Jacobian matrix of the system are, in addition to the three zero eigenvalues, approximately as follows: 
\begin{equation*}
\begin{tabular}{lll}
-163.308657649675, &-68.5596972162577 , &-57.0205793889569 , \\
-16.4435472947534, &-12.1029003142539, & -9.27515541335710, \\ 
-9.27515541335710, &-3.08709626767693 , & -0.209550347944487.
\end{tabular}
\end{equation*}

This steady state is also exponentially stable.  (A third steady state, not shown, is unstable.)
Hence, 
$\mathcal{N}'$ is bistable. 
%
Finally, as $\mathcal{N}'$ is a subnetwork obtained from $\mathcal N$ 
by making some reversible reactions irreversible, then by~\cite[Theorem 3.1]{Joshi-Shiu-2013},  bistability ``lifts'' from $\mathcal{N}'$ to $\mathcal N$.  Thus, $\mathcal N$ is bistable.
\end{proof}

We obtain the following immediate consequence of Theorem~\ref{thm:bistable}.

\begin{corollary} \label{cor:irrev-monostationary}
The fully irreversible ERK network is monostationary.
\end{corollary}

\section{Maximum number of steady states}
\label{sec:max-mumber}
In the previous section, 
we saw that the full ERK network and some irreversible ERK networks (those with $\kon$ or $\lon$) are bistable, admitting two stable steady states in a stoichiometric compatibility class. The question arises, Do these networks admit three or more such steady states?  We suspect not (Conjecture~\ref{conj:number-steady-states}).  

As a step toward resolving this problem, here we investigate the maximum number of positive steady states in ERK networks, together with some related measures we introduce, the maximum number of (non-boundary) complex-number steady states and the ``mixed volume''.  
The mixed volume is always an upper bound on the number of complex steady states (Proposition~\ref{prop:bounds}), but we show these numbers are equal for ERK networks (Proposition~\ref{prop:mixed-vol}).

\subsection{Background and new definitions} \label{sec:MV-bkrd}
Here we recall from~\cite{which-small} a network's maximum number of positive steady states, and then extend the definition to allow for complex-number steady states.

\begin{defn} \label{def:capacity}
A network \defword{admits $k$ positive steady states} 
(for some $k \in \mathbb{Z}_{\geq 0}$)
if there exists a choice of positive rate constants so that the resulting mass-action system~\eqref{sys}
has exactly $k$ positive steady states in some stoichiometric compatibility class~\eqref{eqn:invtPoly}.  
\end{defn}

In~\cite{which-small}, $k=\infty$ was allowed when there are infinitely many steady states in a stoichiometric compatibility class.  Here we do not allow $k=\infty$ so that we consider isolated roots only (as in Proposition~\ref{prop:bernstein} below).

\begin{defn} \label{def:complex-capacity}
Let $G$ be a network with $s$ species, 
$m$ reactions, and 
a $d \times s$ conservation-law matrix $W$, 
which results in the system augmented by conservation laws 
$f_{c,\kappa}$, as in~\eqref{consys}.  
The network $G$ \defword{admits $k$ steady states over $\mathbb{C}^*$} if 
    there exists a choice of positive rate constants 
    $\kappa \in \mathbb{R}^m_{>0}$
and a total-constant vector $c \in \mathbb{R}^d$ such that the system $f_{c,\kappa}=0$ has exactly $k$ solutions in $(\mathbb{C}^*)^s = (\mathbb{C} \setminus \{0\} )^s$.
\end{defn}

It is straightforward to check that Definition~\ref{def:complex-capacity} does not depend on the choice of $W$.

\begin{defn} \label{def:max}
   The \defword{maximum number of positive steady states} 
   (respectively, \defword{maximum number of steady states over $\mathbb{C}^*$})
   of a network $G$ is the maximum value of $k$ for which $G$ admits $k$ positive steady states (respectively, $k$ steady states over $\mathbb{C}^*$).
\end{defn}

Next we recall, from convex geometry, the concept of mixed volume, which we will apply to reaction networks.  For background on convex and polyhedral geometry (such as polytopes and Minkowski sums), we direct the reader to Ziegler's book~\cite{ziegler}. In particular, for 
a polynomial $f = b_1 x^{\sigma_1} + b_2 x^{\sigma_2} + \dots + b_{\ell} x^{\sigma_{\ell}} \in \mathbb{R}[x_1,x_2,\dots, x_s]~$, where 
the 
exponent vectors $\sigma_i \in \mathbb{Z}^s$ are distinct and 
$b_i \neq 0$ for all $i$, 
the \defword{Newton polytope} of $f$
is the convex hull of its exponent vectors:
$	{\rm Newt}(f) \coloneqq {\rm conv} \{\sigma_1,~ \sigma_2,~ \dots~ ,~ \sigma_{\ell}\} ~\subseteq~ \mathbb{R}^s.$

\begin{defn} 
Let $P_1, P_2, \ldots,P_s\subseteq \R^s$ be polytopes. 
The volume of the Minkowski sum $\lambda_1P_1+ \lambda_2 P_2+ \ldots+\lambda_s P_s$ is a homogeneous polynomial of degree $s$ in nonnegative variables $\lambda_1,\lambda_2,\ldots,\lambda_s$. In this polynomial, the coefficient 
of $\lambda_1 \lambda_2\cdots\lambda_s$, 
denoted by $\vol(P_1, P_2, \ldots,P_s)$, is the \defword{mixed volume} of $P_1,P_2,...,P_s$. 
\end{defn}
\noindent
The mixed volume counts the number of solutions in $(\C^*)^s$ of a generic polynomial system.

\begin{proposition}[Bernstein's theorem~\cite{bernstein}] \label{prop:bernstein} 
Consider $s$ real polynomials $g_1, g_2, \dots, g_s \in \mathbb{R}[x_1, x_2, \dots, x_s] $.
Then the number of isolated solutions in $(\C^*)^s$, counted with multiplicity, of the system 
$g_1(x) = g_2(x) = \cdots = g_s(x) = 0$ is at most $\vol(\newt(g_1), \ldots , \newt(g_s))$. 
\end{proposition}

\begin{defn}\label{def:mv-crn}
Let $G$ be a network with $s$ species, 
$m$ reactions, and 
a $d \times s$ conservation-law matrix $W$, 
which results in the system augmented by conservation laws 
$f_{c,\kappa}$, as in~\eqref{consys}.  
Let $c^*\in \mathbb{R}^d_{\neq 0}$, and let 
$\kappa^* \in \mathbb{R}^m_{>0}$ be generic. 
Let $P_1, P_2, \ldots,P_s\subset \R^s$ be the Newton polytopes of $f_{c^*,\kappa^*,1},f_{c^*,\kappa^*,2}, \ldots, f_{c^*,\kappa^*,s}$, respectively. The \defword{mixed volume of $G$} (with respect to $W$)
is the mixed volume of $P_1,P_2,\ldots,P_s$.
\end{defn}
 A closely related definition is introduced and analyzed by Gross and Hill~\cite{gross-hill}.

\begin{remark}\label{ref:mv-well-defined}
The mixed volume (Definition~\ref{def:mv-crn}) is well defined.  Indeed, it is straightforward to check that the exponents appearing in $f_{c^*, \kappa^*}$ are the same as long as $c^* \in \mathbb{R}^d_{\neq 0}$ and $\kappa^*$ is chosen generically (so that no coefficients of $f_{c^*, \kappa^*}$ vanish, or equivalently certain linear combinations of the $\kappa_j$'s do not vanish).
\end{remark}

\subsection{Results} \label{sec:MV-results}
Every positive steady state is a steady state over $\mathbb{C}^*$.
Also, the mixed volume pertains to polynomial systems with the same supports (i.e., the exponents that appear in each polynomial) as the augmented system $f_{c,\kappa}=0$ (but without constraining the coefficients to come from a reaction network).  We obtain, therefore, the bounds in the following result:

\begin{proposition} \label{prop:bounds}
For every network, the following inequalities hold among the
maximum number of positive steady states, 
the maximum number of steady states over $\mathbb{C}^*$, and 
the mixed volume of the network (with respect to any conservation-law matrix):
\[
\mathrm{max~\#~of~positive~steady~states}
~\leq ~ 
\mathrm{max~\#~of~steady~states~over~} \mathbb{C}^*
~\leq ~
{\rm mixed~volume}~.
\]
\end{proposition}
\begin{proof}
This result follows from Proposition~\ref{prop:bernstein} and Definitions~\ref{def:capacity}--\ref{def:max}.
\end{proof}

We investigate the numbers in Proposition~\ref{prop:bounds} for ERK networks in the following result.

\begin{proposition} \label{prop:mixed-vol}
Consider four ERK networks: the full ERK network, the full ERK network with the reaction $\kon$ removed, the fully irreversible network, and the reduced network.
For these networks, the following numbers (or bounds on them) are given in Table~\ref{tab:number}: 
the maximum number of positive steady states, 
the maximum number of steady states over $\mathbb{C}^*$, and
the mixed volume of the network (with respect to the consveration laws 
\eqref{eq:conservation} or~\eqref{eq:cons-law-reduced}).
\end{proposition}

\begin{table}[ht]
\begin{center}
\begin{tabular}{lccc}
\hline
 ERK &  Max \# & Max \# & Mixed\\
network                                & positive steady states & over $\mathbb C^*$ & volume\\
\hline
Full      & $\geq 3$                        & 7    & 7                            \\
Full with $\kon=0$ & $\geq 3$ & 5 & 5 \\ 
Fully irreversible                             & 1                                    & 3  & 3                            \\
Reduced                                    & 1                                    & 3 & 3    \\                       \hline
\end{tabular}
\caption{Results on ERK networks. \label{tab:number} }
\end{center}
\end{table}

\begin{proof}
The results on the mixed volume were computed using the {\tt PHCpack}~\cite{phcpack} package in {\tt Macaulay2}~\cite{M2}. See the supplementary file {\tt ERK-mixedVol.m2}.  

The mixed volume is an upper bound on the maximum number of steady states over $\mathbb{C}^*$ (Proposition~\ref{prop:bounds}), so we need only show that each network admits the number 
shown in Table~\ref{tab:number} for steady states over $\mathbb{C}^*$.

The full ERK network admits 7 steady states over $\mathbb{C}^*$ (including 3 positive steady states)~\cite[Example 3.18]{DPST}.  
Next, we consider the remaining three networks (see the supplementary file {\tt ERK-MaxComplexNumber.nb}).

For the full ERK network with $k_{on}=0$, 
 when $(c_1 , c_2,  c_3) =(1, 2,  3)$ 
 and 
$(k_1, k_2,
    $ 
    $
    	k_3, 
	\kcat, 
	\kon, 
	\koff, 
    	{\ell}_1,
    	{\ell}_2,
	{\ell}_3, 
	\lcat, 
	\lon,
	\loff, 
	m_1,
    	m_2, 
	m_3,
    	n_1, 
    	n_2,
	n_3
    ) 
    =
  (    
  	3, 
  	25,  
    	1, 
	5,  
	0, 
	6, 
    	5, 
    	23, 
	11, 
	13, 
    $ 
    $
	  43, 
	41, 
	12, 
    	7, 
	8, 
    	12, 
    	31, 
	21 
     ),$
we obtain $5$ steady states over $\mathbb{C}^*$, three real and one complex-conjugate pair, which are approximately as follows:
{\footnotesize 
\begin{align*}
\begin{tabular}{llll}
(21.7475,  1.97705,& 2.40601,  2.64849, &
  0.760404,  0.564871,& -24.1306, -0.973762,  \\
  -2.51373, -0.28488, &  -7.81077, -18.495), &&\\
  (5.4105 + 14.8132 i, &
  0.491864 + 1.34665 i,  &
  1.97942 - 3.45492 i,  &
  1.66315 - 1.90055 i, \\
  0.189178 + 0.517943 i, &
  0.140532 + 0.384758 i, &
  -5.88178 - 12.7049 i,  &
  1.00714 + 0.997852 i, \\
 1.13283 + 0.533085 i, &
  0.470121 + 0.662785 i, & 
  -9.72843 - 0.81303 i, &
  -0.749157 - 12.0899 i), \\
  (5.4105 - 14.8132 i, &
 0.491864 - 1.34665 i, & 
 1.97942 + 3.45492 i, &
  1.66315 + 1.90055 i, \\
   0.189178 - 0.517943 i, &
   0.140532 - 0.384758 i, &
   -5.88178 + 12.7049 i, &
  1.00714 - 0.997852 i, \\
  1.13283 - 0.533085 i, &
  0.470121 - 0.662785 i, &
  -9.72843 + 0.81303 i, &
   -0.749157 + 12.0899 i)\\
    (9.63546, 0.875951, &
  -0.488295, 0.0430355, &
 0.336904,  0.250272,  &
  -8.02311,  2.36979,\\
   0.45764,  0.173889, &
   -10.4083,  0.123488), & \text{and} &\\
    ( 0.163415,  0.0148559,&
  0.00111949, 0.00756688,  &
  0.00571382, 
   0.00424455, &
   1.82061, 2.98247, \\
   0.00616705,
  0.00175686,& 0.777908, 0.0172524). &&
\end{tabular}
  \end{align*}}

For the fully irreversible ERK network, 
when $(c_1 , c_2, c_3) =(1, 2,  3)$ and 
$(k_1, 
    	k_3, 
	\kcat, 
    $ 
    $
	\koff, 
    	{\ell}_1,
	{\ell}_3, 
	\lcat, 
	\loff, 
    	m_2, 
	m_3,
    	n_1, 
	n_3
    ) 
    =
     (    
  	3, 
  		1, 
	5,  
	6, 
    	5, 
    11, 
	13, 
	41, 
		7, 
	8, 
    	12, 
	21 
     )$,
there are 3 steady states over $\mathbb{C}^*$, all real, 
with approximate values:
{\footnotesize
\begin{align*}
\begin{tabular}{llllllll}
(14.199, & 1.29082, & 2.5444, &  2.43721, & 
   0.496468,&  0.368805,&  -16.0342,&  -0.302478,\\
  -2.13373, &
   -0.181355, & -0.295181, & 
   -17.7264),& &&&\\
   ( 0.490202,& 0.0445638, & 0.0878422,&
    0.0841415, & 0.0171399, & 0.0127325,&  1.37739, &
  2.88599, \\
  0.00772073, &
  0.0728849, & 0.118631, & 
   0.0641415), &\text{and} &&&\\  
(1.9419, & 0.176536, & 0.34798, & 
   0.33332, &  0.0678986,& 0.050439, &-0.466416,&  
   2.54834, \\
   0.0346375, &
   -0.852654, &    -1.38782, & 
  0.287758). &&&&
  \end{tabular}
   \end{align*}
}  
For the reduced ERK network, 
  let $(c_1 , c_2,  c_3) = (1, 2,  3)$ and 
$(k_1,
    	k_3, 
	\kcat, 
	\koff, 
	m,
	n,
    	{\ell}_1,
    $ 
    $
	{\ell}_3, 
	\lcat, 
	\loff
	) 
   =
    (  	3, 
    	4, 
	1,  
	5, 
	6,8,
    	7, 
    11, 
	12, 
	5 
	 )$. 
We obtain $3$ steady states over $\mathbb{C}^*$, all real, which are approximately:
{\footnotesize
\begin{align*}
\begin{tabular}{llllllll}
( -0.843105, & -37.1185, & 23.4711, &  15.6474, & 
  -9.92245, &-30.6429,&  -0.0292745,& -0.319149, \\
  2.0152, & 1.30395), &&&&&& \\
 (0.314129, & 1.4361, & 
   0.338341, & 0.22556, & 0.015463, & 0.0477534, &
   0.0109073, & 2.95215,  \\
   0.0290494, &  0.0187967), & \text{and} &&&&&\\
(-2.47545,& -0.954967, & 1.77298, &
  1.18199,&  0.087009, & 0.268704,& -0.0859532, & 
  2.74928,\\
  0.152226, & 0.0984989).&&&&&&
  \end{tabular}
   \end{align*} }

Finally, we examine the maximum number of positive steady states.  We already saw that the fully irreversible and reversible networks are monostationary (Corollary~\ref{cor:irrev-monostationary} and Proposition~\ref{prop:bistable}, respectively).  
For the ``partially irreversible'' network, we saw in the proof of Theorem~\ref{thm:bistable} that it admits 3 positive steady states.  As for the full network, as noted above, 3 positive steady states were shown in~\cite[Example 3.18]{DPST}. 
\end{proof}

Table~\ref{tab:number} suggests that the mixed volume is a measure of the complexity of a network.  The full ERK network is multistationary, and its mixed volume is 7. The mixed volume drops to 5 when $\kon=0$. When the network is further simplified to the fully irreversible, or even to the reduced ERK network, the mixed volume becomes 3, and bistability is lost as well. 

Finally, we conjecture that the bounds in Table~\ref{tab:number} 
are strict, and ask about stability.
\begin{conjecture} \label{conj:number-steady-states}
For the full ERK network and the full ERK network with $\kon = 0$, the maximum number of positive (respectively, positive stable) steady states is 3 (respectively, 2).
\end{conjecture}


\section{Discussion} \label{sec:discussion}
Phosphorylation plays a key role in cellular signaling networks, such as 
{\em mitogen-activated protein kinase (MAPK) cascades}, which enable cells to make decisions (to differentiate, proliferate, die, and so on)~\cite{chang-karin}.  
This decision-making role of MAPK cascades suggests that they exhibit switch-like behavior, i.e., bistability.  Indeed, bistability in such cascades has been seen in experiments \cite{bistab-cascade, mapk-bistable}.  Oscillations also have been observed~\cite{yeast-mapk-oscillations,oscillations-mapk-cancer}, hinting at a role in timekeeping. Indeed, multisite phosphorylation is the main mechanism for establishing the 24-hour period in eukaryotic circadian clocks~\cite{ode,beat}.

These experimental findings motivated the questions we pursued.  Specifically, we investigated robustness of oscillations and bistability in models of ERK regulation by dual-site phosphorylation.  
Bistability, we found, is quickly lost when reactions are made irreversible.  
Indeed, bistability is characterized by the presence of two specific reactions.  
Oscillations, in contrast, persist even as the network is greatly simplified.  Indeed, we discovered oscillations in the reduced ERK network. 
Moreover, this network has 
the same number of reactions (ten) as the mixed-mechanism network which Suwanmajo and Krishnan surmised ``could be the simplest enzymatic modification scheme that can intrinsically exhibit oscillation''~\cite[\S 3.1]{SK}.  Our reduced ERK network, therefore, may also be such a minimal oscillatory network.

Returning to our bistability criterion (Theorem~\ref{thm:bistable}),
recall that this result elucidates which reactions can be safely {\em deleted}
while preserving bistability -- in contrast to standard results 
concerning reactions that can be {\em added}~\cite{BP-inher, FeliuWiuf, Joshi-Shiu-2013}.  
We desire more results of this type, so we comment  
on how we proved our result.  The key was the special form of the steady-state parametrization.
In particular, following \cite{DPST}, our parametrizations allow both species concentrations and rate constants to be solved (at steady state) in terms of other variables.  
Additionally, a single parametrizations 
specialized (by setting rates to zero for deleted reactions) 
to obtain parametrizations for a whole family of networks.
Together, these properties gave us access to new information on how bistability is controlled. 
We are interested, therefore, in the following question: {\em Which networks admit a steady-state parametrization that specializes for irreversible versions of the network?}  

Our results on oscillations were enabled by new mathematical approaches to find Hopf bifurcations.  Specifically, building on~\cite{mixed}, we gave a Hopf-bifurcation criterion for networks admitting a steady-state parametrization.  Additionally, we successfully applied this criterion to the reduced ERK network by analyzing the Newton polytopes of certain Hurwitz determinants.  We expect these techniques to apply to more networks.

Finally, our work generated a number of open questions.  First,
what are the mixed volumes of irreversible versions of the ERK network (beyond those shown in Table~\ref{tab:number})?  
In particular, is there a mixed-volume analogue of our bistability criterion, which is in terms of the reactions $\kon$ and $\lon$? 
And, what is the maximum number of (stable) steady states in the full ERK network (Conjecture~\ref{conj:number-steady-states})?
Progress toward these questions will yield further insight into robustness of bistability and oscillations in biological signaling networks.




\subsection*{Acknowledgements}
{\footnotesize
NO, AS, and XT were partially supported by the NSF (DMS-1752672).
AT was partially supported by the Independent Research Fund Denmark. 
The authors thank 
Elisenda Feliu for insightful comments on an earlier draft, 
and thank
Carsten Conradi, 
Elizabeth Gross, 
Cvetelina Hill, 
Maya Mincheva, 
Stanislav Shvartsman, 
Frank Sottile, 
Elise Walker, 
and Timo de Wolff for helpful discussions.
This project was initiated while AT was a visiting scholar at Texas A\&M University, and while XT was hosted by ICERM. We thank Texas A\&M University and ICERM for their hospitality.
}

\appendix
\section{Files in the Supporting Information} \label{app:supp-files}
Table~\ref{tab:supp-files} lists the 
files in the Supporting Information, and the result/proof each file supports. 
All files can be found at the online repository: \url{https://github.com/neeedz/ERK}.

\begin{table}[ht]
\centering
\begin{tabular}{lll}
\hline
Name & File type & Result   \\
\hline
{\tt ERK-Matcont.txt}     &         text file with {\tt MATCONT} instructions
    & Figures~\ref{fig:oscill-in-irr-erk} and~\ref{fig:oscill-in-red-erk}       \\
{\tt irreversibleERK.mw}     &         {\tt Maple}
    & Theorem~\ref{thm:bistable}         \\
{\tt reducedERK-noMSS.mw}     &         {\tt Maple}
    &  Proposition~\ref{prop:bistable}        \\
{\tt reducedERK-hopf.mw}     &         {\tt Maple}
    &   Theorem~\ref{thm:hopf}       \\
{\tt reducedERK-cones.sws}     &         {\tt Sage}
    &  Theorem~\ref{thm:hopf}        \\
{\tt ERK-mixedVol.m2}     &         {\tt PHCPack}
    &  Proposition~\ref{prop:mixed-vol}        \\
{\tt ERK-MaxComplexNumber.nb}   &         {\tt Mathematica}                 
    & Proposition~\ref{prop:mixed-vol} \\
\hline
\end{tabular}
\caption{Supporting Information files and the results they support. \label{tab:supp-files}}
\end{table}

\section{Newton-polytope method}\label{app:cones}
Here we show how analyzing the Newton polytopes of two polynomials can reveal whether there is a positive point at which one polynomial is positive and simultaneously the other is zero (Proposition~\ref{app_prop:cones_method} and Algorithm~\ref{alg:cones}).  In Appendix~\ref{sec:appendix-using-cones-method}, we show how we used this approach, which we call the Newton-polytope method, to find a Hopf bifurcation leading to oscillations in the reduced ERK network (in Theorem~\ref{thm:hopf}).

\begin{notation} \label{notation:NP}
Consider a polynomial $f = b_1 x^{\sigma_1} + b_2 x^{\sigma_2} + \dots + b_{\ell} x^{\sigma_{\ell}} \in \mathbb{R}[x_1,x_2,\dots, x_s]$, where 
the exponent vectors $\sigma_i \in \mathbb{Z}_{\geq 0}^s$ 
are distinct and $b_i \neq 0$ for all $i$.
A vertex $\sigma_i$ of $\newt(f)$, 
the Newton polytope of $f$, 
is a \defword{positive vertex} (respectively, \defword{negative vertex}) 
if the corresponding monomial of $f$ is positive, i.e., $b_i >0$ (respectively, $b_i<0$). 
Also, $N_{f}(\sigma)$ denotes the \defword{outer normal cone of the vertex $\sigma$} of $\newt(f)$, i.e., the cone generated by the outer normal vectors to all supporting hyperplanes of $\newt(f)$ containing the vertex $\sigma$.  Finally, for a cone $C$, let $\mathrm{int}(C)$ denote the relative interior of the cone.
\end{notation}

For an extensive discussion on polytopes and normal cones, see \cite{ziegler}.

\begin{proposition} 
\label{app_prop:cones_method}
Let $f,g \in \R[x_1,x_2,\ldots x_s]$. 
Assume that $\alpha$ is a positive vertex of ${\rm Newt}(f)$, 
$\beta_+$ is a positive vertex of ${\rm Newt}(g)$, and $\beta_-$ is a negative vertex of 
${\rm Newt}(g)$.
Then, if 
$\mathrm{int} (N_{f}(\alpha)) \cap \mathrm{int} (N_g(\beta_+)) $
and 
$\mathrm{int}(N_{f}(\alpha)) \cap \mathrm{int} (N_g(\beta_-))$ are both nonempty,
then there exists $x^* \in \mathbb{R}^s_{>0}$ such that $f(x^*)>0$ and $g(x^*)=0$.
\end{proposition}

To prove Proposition~\ref{app_prop:cones_method} we use the following well-known lemma and its proof. 

\begin{lemma} \label{lem:NP}
For a
real, multivariate polynomial 
$	f = b_1 x^{\sigma_1} + b_2 x^{\sigma_2} + \dots + b_{\ell} x^{\sigma_{\ell}} 
		\in \mathbb{R}[x_1,x_2,\dots, x_s],$
if $\sigma_i$ is a positive vertex (respectively, negative vertex) of ${\rm Newt}(f)$,
then  there exists $x^* \in \mathbb{R}^s_{>0}$ such that $f(x^*)>0$ (respectively, $f(x^*)<0$).
\end{lemma}

\begin{proof}
Let $\sigma_i$ be a vertex of ${\rm Newt}(f)$.
Pick $w=(w_1,w_2,\ldots,w_s)$ in the relative interior of the outer normal cone $ N_f(\sigma_i)$, which exists because $\sigma_i$ is a vertex.
Then, by construction, the linear functional 
$\langle w, - \rangle$ is maximized over the exponent-vectors $\sigma_1, \sigma_2, \dots, \sigma_{\ell}$ at $\sigma_i$.
Thus, we have the following univariate ``polynomial with real exponents'' in $t$:
\[
f(t^{w_1},t^{w_2},\ldots,t^{w_s})
 ~=~
 b_1 t^{\langle w, \sigma_1\rangle} + 
 b_2 t^{\langle w, \sigma_1\rangle} + 
    \dots + 
 b_{\ell} t^{\langle w, \sigma_{\ell} \rangle}
 ~=~
 b_i t^{\langle w, \sigma_i \rangle} + \textrm{lower-order~terms}~.
\]
%
So, for $t$ large, 
$\sign(f(t^{w_1},t^{w_2},\ldots,t^{w_s}))=\sign(b_i)$. 
Note that $(t^{w_1},t^{w_2},\ldots,t^{w_s})\in \R^s_{>0}$. 
\end{proof}

Our proof of Proposition~\ref{app_prop:cones_method} is constructive, through the following algorithm, where we use the notation
$f_w(t):=f(t^{w_1}, t^{w_2},\ldots,t^{w_s})$, for $t\in \R$ and $w=(w_1, w_2, \dots, w_s)\in\R^s$.
\begin{algorithm}[H] \label{alg:cones}
\SetKwInOut{Input}{input}\SetKwInOut{Output}{output}
\SetKwInOut{Return}{return}

\Input{polynomials $f,g$, and vertices $\alpha, \beta_+, \beta_-$, as in Proposition~\ref{app_prop:cones_method}
}
\Output{a point $x^*\in \R^s_{>0}$ s.t. $f(x^*)>0$ and $g(x^*)=0$}
 define $C_0\coloneqq  {\rm int}(N_{f}(\alpha)) \cap {\rm int}( N_g(\beta_+))$ and $C_1:={\rm int}(N_{f}(\alpha)) \cap {\rm int}(N_g(\beta_-))$\;
 pick $\ell =(\ell_1,\ell_2,\ldots,\ell_s)\in C_0$ and $m=(m_1,m_2,\ldots,m_s)\in C_1$\;
 define $f_{\ell}(t):=f(t^{\ell_{1}},t^{\ell_{2}},\ldots,t^{\ell_{s}})$; define $f_{m}(t)$; define $g_{\ell}(t)$; define $g_m(t)$\;
 define $\tau_\ell:=\inf\{t^*\in\R_{> 0}\ |\ f_{\ell}(t)> 0\mathrm{\ and\ } g_{\ell}(t)> 0 \mbox{ for all } t>t^*$\}\;
 define $\tau_m:=\inf\{t^*\in\R_{> 0}\ |\ f_{m}(t)> 0\mathrm{\ and\ } g_{m}(t)< 0 \mbox{ for all } t>t^*\}$\;
 define $T:=\max \{\tau_\ell,\tau_m\}+1$\; 
 define $h(r):=f_{r \ell+(1-r) m}(T)$\;
 \While{$\min\{h(r)\ |\ r\in [0,1]\}\leq 0$}{
  $T:=2T$\;
  $h(r):=f_{r \ell+(1-r) m}(T)$\;
 }
 define $r^*:=\argmin\{\left(g_{r\cdot \ell +(1-r)m}(T)\right)^2\ |\ r\in[0,1]\}$  (pick one $r^*$ if there are multiple)\;
 \Return{$T^{r^*\ell +(1-r^*)m} 
 := \left(T^{r^*\ell_1 +(1-r^*)m_1}, T^{r^*\ell_2 +(1-r^*)m_2}, \ldots , T^{r^*\ell_s +(1-r^*)m_s}\right)$
 }
 \caption{Newton-polytope method}
\end{algorithm}

\begin{proof}[Proof of Proposition~\ref{app_prop:cones_method}]
Let $a_+x^{\alpha}$ be the term of $f$ corresponding to the vertex $\alpha$ of $\newt(f)$, and similarly 
let $b_+x^{\beta_+}$ (respectively, $b_-x^{\beta_-}$)
be the term of $g$ corresponding to the vertex $\beta_+$ 
(respectively, $\beta_-$)
of $\newt(g)$. Thus, $a_+>0$, $b_+>0$, and $b_-<0$.
Let $\{a_1,a_2, \ldots ,a_{d}\}\subseteq \mathbb{R}$ denote the remaining set of coefficients of $f$, so that $f= a_+ x^{\alpha} + 
(a_1 x^{\sigma_1}+  a_2 x^{\sigma_2}+ \dots + a_d x^{\sigma_d})$, for some exponent vectors $\sigma_i \in \mathbb{Z}_{\geq 0}^s$.

\underline{Algorithm~\ref{alg:cones} terminates}: 
First, $\ell$ and $m$ in line 2 exist by hypothesis.
Also, $\tau_{\ell}$ and $\tau_m$ in lines 4--5 exist by the proof of Lemma~\ref{lem:NP} and by construction.
Next, $\min h(r)$ in line 8 exists because $h$ is a continuous univariate function defined on a compact interval. 

By construction and because cones are convex, the vector $r \ell+(1-r) m$, 
which is a convex combination of $\ell$ and $m$,
is in the relative interior of $N_{f}(\alpha)$ 
 for all $r\in [0,1]$.  
Thus, $\langle r \ell + (1-r) m, ~\alpha - \sigma_i \rangle >0$ for all $i=1,2,\dots, d$ and for all $r \in [0,1]$.  This (together with a straightforward argument using continuity and compactness) implies the following:
\[
    \delta~:=~
    \inf_{r \in [0,1]} \min_{i=1,2,\dots, d}
    \langle r \ell + (1-r) m, ~ \alpha - \sigma_i \rangle 
    ~>~0.
\]
Next, let $\beta:= \inf_{r\in[0,1]} \langle r \ell + (1-r) m, ~ \alpha  \rangle$.  
Then, for all $r \in [0,1]$ and $t>0$, 
\begin{align} \label{eq:bound} \notag
    f_{r \ell+(1-r) m}(t) ~&=~
    a_+ t^{\langle r \ell+(1-r) m), \alpha \rangle}
    +
    \left( a_1 t^{\langle r \ell+(1-r) m), \sigma_1 \rangle}
    +
    \dots +
    a_d t^{\langle r \ell+(1-r) m), \sigma_d \rangle} \right) \\
    ~&>~
    a_+ t^{\beta} 
    -
    (|a_1|+|a_2|+ \dots + |a_d|) t^{\beta - \delta} ~=:~ \widetilde{f}(t)~.
\end{align}
In $\widetilde{f}(t)$, 
the term $a_+ t^{\beta}$ dominates the other term, for $t $ large, 
so there exists $T^*>0$ such that 
$\widetilde{f}(t)\geq 0$ when $t \geq T^*$.  
So, by~\eqref{eq:bound}, the while loop in line 8 ends when $T\geq T^*$ (or earlier). 



\underline{Algorithm~\ref{alg:cones} is correct}: 
For $T$ fixed, the minimum of $\psi(r):=\left(g(T^{r\ell+(1-r)m})\right)^2$ over 
the compact set $[0,1]$ is attained, because $\psi$ is continuous.  Next we show that 
this minimum value is 0, or equivalently that for 
$\chi(r) := g(T^{r \ell +(1-r)m})$ 
there exists some $r^*\in (0,1)$ such that 
$\chi(r^*) 
=0$. 
Indeed, this follows from the Intermediate Value Theorem, 
because $\chi$ is continuous,  
$\chi(0) 
=g(T^m)<0$ (because $T > \tau_m$), 
and
$\chi(1) 
=g(T^{\ell})>0$ (because $T>\tau_{\ell})$. 

Finally, the inequality $f(T^{r^*\ell +(1-r^*)m})>0$ holds by construction of $T$, so defining 
$x^*:=T^{r^*\ell +(1-r^*)m} \in \R^s_{>0}$ yields the desired vector satisfying $f(x^*)>0$ and $g(x^*)=0$.
\end{proof}

\section{Using the Newton-polytope method} \label{sec:appendix-using-cones-method}
Here we show how we used Algorithm~\ref{alg:cones} to find the Hopf bifurcation in Theorem~\ref{thm:hopf}.  
(For details, 
see the supplementary files {\tt reducedERK-hopf.mw} and {\tt reducedERK-cones.sws}).
Recall from the proof of that theorem, that our goal was to find some $x^* \in \mathbb{R}^{10}_{>0}$ and $\hat{\kappa}^* = (\kcat^*, \koff^*, \loff^*) \in \mathbb{R}^3_{>0}$ satisfying the following conditions from Proposition~\ref{prop:hopfcriterion}:
    \begin{align} \label{eq:hurwitz-hopf-conditions-reduced-ERK-again}
   \mathfrak{h}_4 (\hat{\kappa}^*; x^*) >&0~, ~ 
    \mathfrak{h}_5 (\hat{\kappa}^*; x^*) >0~, ~ 
    \mathfrak{h}_6 (\hat{\kappa}^*; x^*) =0~, ~{\rm and} ~ 
    \frac{\partial}{\partial \kcat }& \mathfrak{h}_6 (\hat{\kappa}; x) |_{(\hat{\kappa}; x)=(\hat{\kappa}^*; x^*)} \neq 0~.
    \end{align}

{\bf Step One.}
Specialize some of the parameters: set $k_{\rm off}=\ell_{\rm off}=1$ and $x_3=x_4=x_6 = x_7= x_9= x_{10} = 1$.    (Otherwise, $\mathfrak{h}_5$ and $\mathfrak{h}_6$ are too large to be computed.)

{\bf Step Two.}
Do a change of variables: let 
$y_i = 1/x_i$ for $i=1,2,5,8$.
These variables $x_i$ were in the denominator, so switching to the variables $y_i$ yield polynomials.  

Let 
$\mathcal{H}_4$, $\mathcal{H}_5$, and $\mathcal{H}_6$ denote the resulting polynomials in $\mathbb{Q}[\kcat, y_1, y_2, y_5, y_8]$ after performing Steps One and Two.  Accordingly, our updated goal is to find $(\kcat^*, y^*_1, y^*_2, y^*_5, y^*_8) \in \mathbb{R}^{5}_{>0}$ at which $\mathcal{H}_4$ and $\mathcal{H}_5$
are positive and $\mathcal{H}_6$ is zero.  (In a later step, we must also check the partial-derivative condition in~\eqref{eq:hurwitz-hopf-conditions-reduced-ERK-again}.)

{\bf Step Three.}
Apply (a straightforward generalization of)
Algorithm~\ref{alg:cones} as follows.
\begin{enumerate}[(i)]
    \item Find a positive vertex of $\mathcal{H}_4$ and a positive vertex of $\mathcal{H}_5$
    whose outer normal cones intersect (denote the intersection by $C$), and a positive vertex and a negative vertex of $\mathcal{H}_6$ (denote their outer normal cones by $D_+$ and $D_-$, respectively) for which:
        \begin{enumerate}
            \item the intersection $D_+ \cap D_-$ is 4-dimensional, and 
            \item the intersections $C \cap D_+$ and $C \cap D_-$ are both 5-dimensional.
        \end{enumerate}
    \item By Proposition~\ref{app_prop:cones_method}, a vector $(\kcat^*, y^*_1, y^*_2, y^*_5, y^*_8)$ that accomplishes our updated goal, is guaranteed.  To find such a point, we follow Algorithm~\ref{alg:cones} to obtain 
     $\kcat^*=729,y_1^*\approx 16.79978292, y_2^*\approx 453.5941389, y_5^*\approx 6.587368051$, and $y_8^*\approx 80675.77181$.
 \end{enumerate}

Recall the specializations in Step One and change of variables in Step Two,
to obtain 
$\hat{\kappa}=(729,1,1)$ and  
\[
   x^*~\approx~ (
0.05952457867,~ 
0.002204614024,~ 
1,~ 
1,~ 
0.1518056972,~ 
1,~ 
1,~ 
0.00001239529511,~ 
1,~ 
1 
)~.\]


{\bf Step Four.}
Verify that the conditions in~\eqref{eq:hurwitz-hopf-conditions-reduced-ERK-again} hold. 

\end{document}